%% file: generic-color.tex
\documentclass[journal,twoside,web]{ieeecolor}
\usepackage{generic}
\usepackage{amsmath,amssymb,amsfonts}
\usepackage{algorithmic}
\usepackage{textcomp}

\usepackage{graphicx,color}
\usepackage{fancyhdr}
\usepackage{epsfig} 
\usepackage{psfrag}
\usepackage{cite}
\usepackage{dsfont}  
\usepackage{todonotes}
\usepackage{tikz}
\usetikzlibrary{arrows,shapes,automata,backgrounds,petri,calc}

\usepackage{booktabs}
\usepackage[export]{adjustbox}
\usepackage{bbm}

\usepackage{caption}
\usepackage{subcaption}
\DeclareMathOperator{\E}{\mathsf{E}}

\DeclareMathOperator{\Tr}{\textsf{Tr}}
\DeclareMathOperator*{\argmin}{\arg\!\min}

\newtheorem{theorem}{Theorem}
\newtheorem{remark}{Remark}

\newtheorem{corollary}{Corollary}

\def\BibTeX{{\rm B\kern-.05em{\sc i\kern-.025em b}\kern-.08em
    T\kern-.1667em\lower.7ex\hbox{E}\kern-.125emX}}
\begin{document}
\title{Delay-sensitive Joint Optimal Control \\and Resource Management in Multi-loop \\Networked Control Systems}
\author{Mohammad H. Mamduhi, Dipankar Maity, \IEEEmembership{Member, IEEE}, Sandra Hirche, \IEEEmembership{Senior Member, IEEE}, John S. Baras, \IEEEmembership{Life Fellow, IEEE}, and Karl H. Johansson, \IEEEmembership{Fellow, IEEE}
\thanks{02.07.2020. 
``This work was supported by the Knut and Alice Wallenberg Foundation, the Swedish Strategic Research Foundation, and the Swedish Research Council.'' 
}
\thanks{M. H. Mamduhi, J. S. Baras and K. H. Johansson are with the Division of Decision and Control Systems, KTH Royal Institute of Technology, Malvinas v\"ag 10, 10044, Stockholm, Sweden (e-mail: \{mamduhi,baras,kallej\}@kth.se). }
\thanks{D. Maity is with the Guggenheim School of Aerospace Engineering, Georgia Institute of Technology, 270 Ferst Dr. Atlanta, GA 30332-0150, USA (e-mail: dmaity@gatech.edu).}
\thanks{S. Hirche is with the Chair of Information-oriented Control, Technical University of Munich, Theresienstr. 90, 80333 Munich, Germany (e-mail: hirche@tum.de).}
\thanks{J. S. Baras is also with the Institute of Systems Research, The University of Maryland, 20742, College Park, MD, USA (e-mail: baras@umd.edu).}
}

\maketitle

\input{abstract}


\input{introduction}
\input{problem}
\input{numerics}


\input{conclusions}

\bibliographystyle{ieeetr}
\bibliography{TCNS}

\appendices
\input{appendix}


%


\end{document}

%% file: abstract.tex
\begin{abstract}
In the operation of networked control systems, where multiple processes share a resource-limited and time-varying cost-sensitive network, communication delay is inevitable and primarily influenced by, first, the control systems deploying intermittent sensor sampling to reduce the communication cost by restricting non-urgent transmissions, and second, the network performing resource management to minimize excessive traffic and eventually data loss. 
In a heterogeneous scenario, where control systems may tolerate only specific levels of sensor-to-controller latency, delay sensitivities need to be considered in the design of control and network policies to achieve the desired performance guarantees. We propose a cross-layer optimal co-design of control, sampling and resource management policies for an NCS consisting of multiple stochastic linear time-invariant systems which close their sensor-to-controller loops over a shared network. Aligned with advanced communication technology, we assume that the network offers a range of latency-varying transmission services for given prices. Local samplers decide either to pay higher cost to access a low-latency channel, or to delay sending a state sample at a reduced price. A resource manager residing in the network data-link layer arbitrates channel access and re-allocates resources if link capacities are exceeded. The performance of the local closed-loop systems is measured by a combination of linear-quadratic Gaussian cost and a suitable communication cost, and the overall objective is to minimize a defined \textit{social cost} by all three policy makers. We derive optimal control, sampling and resource allocation policies under different cross-layer awareness models, including constant and time-varying parameters, and show that higher awareness generally leads to performance enhancement at the expense of higher computational complexity. This trade-off is shown to be a key feature to select the proper interaction structure for the co-design architecture. 
\end{abstract}

%% file: introduction.tex
\section{MOTIVATION And INTRODUCTION}
The design and operation of networked control systems (NCSs), wherein multiple control loops exchange information between their sensors, controllers and actuators via a common communication network, requires a major rethinking to respond to the growing requirements from current and future applications. 
The introduction of communication technologies that provide demand-driven serviceability with adjustable parameters and prices, together with novel approaches to virtually program  network functions and adaptable network features, have created a significant potential to bring control and networking architectures to a whole new level \cite{MOLINA2018407,BORDEL2017156}. This generally means moving from the traditional throughput-oriented and latency-minimizing data transmission with asymptotic-type performance guarantees, to smart data coordination schemes that consider real-time requirements and limitations of both the service providers and service recipients. 

In the context of NCSs, this calls for novel sampling, control and resource management architectures that incorporate the wide range of opportunities provided by the network infrastructure, such as computational capability, adaptive service allocation, virtual programmability, adjustable channel reliability and latency, to maximize quality-of-control (QoC), while minimizing the cost of network usage. Emerging NCS applications, such as networked cyber-physical systems (Net-CPS), Internet of things (IoT), autonomous driving and Industry 4.0, often involve a large number of networked entities, each with time-varying requirements to fulfill specific tasks. The concept of ``network'' in such systems has gone beyond a simple shared communication channel to a general representation of evolving inter-layer dependencies (physical, information, and communication layers) \cite{Baras2014ISCCSP}. This creates a large potential to develop novel interactive approaches for real-time distributed sampling, networking and control in a cross-layer fashion, such that the individual entities become aware of networking architecture and opportunities, and coupling constraints and incorporate them in decision making, while the network is also aware of the demands and the task criticality of the entities and optimally allocate services and adjust the inter-dependencies.
 
\subsection{Contributions}
In this article, we propose jointly optimal communication and control policies for a general NCS model consisting of multiple delay-sensitive heterogeneous stochastic control systems closing their sensor-to-controller loops via a shared communication network, under various inter-layer awareness assumptions. Each sub-system is controlled by two local decision makers: a delay-sensitive controller that determines how fast state information should be sent to the plant controller, and a plant controller that maximizes control performance, measured by a linear-quadratic-Gaussian (LQG) cost. Local controllers have access to partial information of their own loop and may have some knowledge of the network parameters but do not have any knowledge about the dynamics and objectives of other sub-systems. The communication network offers various transmission services, for fixed prices, through multiple capacity-limited channels each with a distinct and deterministic latency. Transmission requests from sub-systems are arbitrated by a resource manager to avoid exceeding the link capacities. Resource arbitration is optimally performed such that the average sum of local (sub-system) LQG cost functions undergoes the minimum deviation compared to the resource-unlimited case, over a finite time horizon. We study scenarios each entailing a specific class of inter-layer awareness (one-directional and bi-directional awareness of time-varying and constant parameters) among the three decision makers, and derive the resulting jointly optimal policies. We show that performance of the joint design is associated with the level of delay-sensitivity tolerances and the awareness structure. In general, higher awareness results in lower local and social costs, though the resulting optimization problem becomes more computationally complex.
We also observe that the extent of performance improvement is firmly tied to the particular model of awareness, that is, for specific scenarios the improvements are slight compared to the extra solution complexity, while for others, the improvements are considerable. Interestingly, stricter delay sensitivity (i.e., local sub-systems tolerate minor deviations from their delay requirements) may result in lower local cost for some specific sub-systems, but higher social cost. 
 Our major contributions in this article are:
\begin{enumerate}
\item introducing a general model of NCS including heterogeneous control loops and variety of network services, with evolving interactions between control and network layers leading to enhanced joint performance.
\item investigating various awareness models for control and network layers and studying the interaction effects on the structure and performance of the optimal co-design.
\item deriving jointly optimal policies from awareness-based social optimization problems including performance-complexity comparisons w.r.t. the awareness model.
\end{enumerate}

We addressed a similar problem for a \textit{single-loop} NCS in \cite{8405590}, however, the present problem is far more general. The setup in \cite{8405590} does not include resource management as no contention exists, and interactions between control and network layers, in the previous formulation, reduce to one-directional knowledge of the network service prices.

\subsection{Related works}
The problem of joint control and communication design in NCSs has been an active research topic for the last two decades in both control and communication communities \cite{Baillieul2007IEEEProc, Shakkottai2003Comm}. Two rather distinct perspectives in addressing it have evolved: from the communication perspective where maximizing quality-of-service (QoS) is the major objective, requirements of control systems are often abstracted in the form of transmission rate, delay, and packet loss, with less attention given to the application dynamics and their real-time necessities \cite{Bai2012ICCPS,s8021099}. Numerous design methodologies are proposed including protocols for QoS-enhacing medium access control (MAC) \cite{YIGITEL20111982,Rajandekar2015ITJ,Bi2013MNA}, resource allocation  \cite{Letaief2006IWC,1561930}, link scheduling and routing \cite{4399978,7479131}, and queuing management \cite{6614116,6115208}. On the other hand, from the control perspective  the aim is to maximize QoC, and the communication network is usually seen as one or more maximum-rate and delay-negligible single-hop channels with some resource management capabilities to resolve contention. Many design approaches for sampling, estimation and control over shared networks are proposed to enhance QoC while reducing the rate of transmission, including event-triggered schemes \cite{FORNI2014490,5510124,SEYBOTH2015392,maity2019optimal}, self-triggered schemes \cite{6425820,7348666}, and adaptive/predictive data transmission and control models \cite{7423697,7798733,6882817}. For more sophisticated models of communication networks with data loss, delay and resource constraints, attempts have been made mostly on co-design architectures that guarantee stability rather than optimality\cite{5409530,7039815,8039513}. Altogether, the efforts have often led to design frameworks that either consider no evolving cross-layer coupling or presume interactions in average form over time, with performance guarantees mostly valid in the asymptotic regime.

New NCS applications, however, include multitude of heterogeneous systems that need to fulfill various real-time tasks while the network is responsible for coordinating the required type of communication and computation services \textit{per-time}. This urges the development of cross-layer architectures that consider active interactions between distributed components of control and communication layers to be aware of each others conditions, capabilities, requirements, and limitations to achieve joint optimal quality-of-control-and-service, not only asymptotically but also over finite time horizons. To achieve this, a main issue to address is optimal timeliness, i.e., when is the best time to make a specific action such as sampling, transmission or actuation. This problem is addressed in the control community mainly for data sampling over single-service communication support leading to optimal event-based technique to restrict unnecessary transmission \cite{7497849,6228518,Molin2014}, and prioritized MAC protocols to distribute resources based on urgency \cite{7963084,MAMDUHI2017209}. These approaches consider some measured or observed quantity of the control system, such as estimation error, as the triggering function. For multiple-loop non-scalar NCS, though, finding the optimal triggering law without major simplifications of the network layer is challenging. Moreover, resource allocation is often performed randomly or based on apriori given parameters but not based on dynamic awareness of interacting layers. In addition, the resulting performances of the proposed approaches are often addressed asymptotically over infinite horizon. 
To the best of our knowledge, a systematic approach that proposes a cross-layer optimal design of control, sampling and resource management strategies to maximize QoC for multi-loop NCSs with a shared network of various service opportunities is not presented in the literature.


\subsection{Notations}
We denote expectation, conditional expectation, transpose, floor and trace operators by $\E[\cdot]$, $\E[\cdot|\cdot]$, $[\cdot]^\top$, $\lfloor\cdot\rfloor$ and $\Tr(\cdot)$, respectively. For $a\geq 0$, define the indicator $\mathbbm{1}(a)\!=\!0$ if $a\!=\!0$,  and $\mathbbm{1}(a)\!=\!1$ if $a\!>\!0$. $X \!\sim \!\mathcal{N}(\mu, W)$ represents a multivariate Gaussian distributed random vector~$X$ with mean vector~$\mu$ and covariance $W \!\succ \!0$, where $A\!\succ \!B$ denotes $A\!-\!B$ is positive definite. The $Q$-weighted squared 2-norm of a column vector $X$ is denoted by $\|X\|^2_{Q}\!\triangleq\!X^\top Q X$. A time-varying column vector $X_t^i$ includes an array of variables belonging to sub-system $i$ at time~$t$, while we define 
$X_{[t_1, t_2]}^i\!\triangleq\!\{X_{t_1}^i,X_{t_1+1}^i,...,X_{t_2-1}^i,X_{t_2}^i\}$, and 
$X^i\!\triangleq\!\{X_{0}^i,X_{1}^i,...\;\}$. 

%% file: problem.tex
\section{NCS Model: Control \& Communication Layers}\label{prob_model}
\begin{figure}[t]
\centering
\psfrag{a}[c][c]{\tiny \text{LTI plant}}
\psfrag{b}[c][c]{\tiny \text{plant controller}}
\psfrag{e}[c][c]{\tiny \text{delay controller}}
\psfrag{t}[c][c]{\scriptsize \text{sub-system} $1$}
\psfrag{tt}[c][c]{\scriptsize \text{sub-system} $N$}
\psfrag{d}[c][c]{\tiny $u_{k}^1$}
\psfrag{dd}[c][c]{\tiny $u_{k}^N$}
\psfrag{g}[c][c]{\tiny $\theta_{k}^1$}
\psfrag{h}[c][c]{\tiny $\theta_{k}^N$}
\psfrag{gg}[c][c]{\tiny $\vartheta_{k}$}
\psfrag{x}[c][c]{\tiny $x_{k}^1$}
\psfrag{xx}[c][c]{\tiny $x_{k}^N$}
\psfrag{f}[c][c]{\scriptsize \text{centralized network manager}}
\psfrag{sss}[c][c]{\scriptsize \text{Ack signal}}
\psfrag{m}[c][c]{\tiny $\lambda_0$}
\psfrag{mm}[c][c]{\tiny $\lambda_1$}
\psfrag{n}[c][c]{\tiny $\lambda_2$}
\psfrag{nn}[c][c]{\scriptsize $\ldots$}
\psfrag{s}[c][c]{\tiny $\vartheta_{k}^1$}
\psfrag{ss}[c][c]{\tiny $\vartheta_{k}^N$}
\psfrag{ddd}[c][c]{ $\ldots$}
\psfrag{o}[c][c]{\tiny $\lambda_D$}
\psfrag{p}[c][c]{\tiny $Z^{-2}(x_{k}^N)$}
\psfrag{q}[c][c]{\tiny $Z^{-1}(x_{k}^1)$}
\psfrag{pp}[c][c]{\tiny $\vartheta_{k}^N, x_{k}^N$}
\psfrag{qq}[c][c]{\tiny $\vartheta_{k}^1, x_{k}^1$}
\includegraphics[width=8.6cm, height=5.6cm]{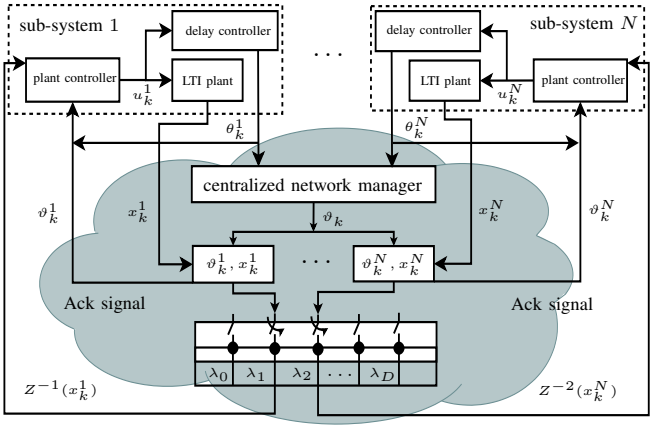}
\caption{Multiple LTI control loops exchange information with their respective controllers over a shared resource-limited communication network that can offer an array of latency-varying transmission services for different prices. ($Z^{-d}$ is the delay operator).}
\label{fig:sys-model}
\vspace{-7mm}
\end{figure}

We consider an NCS consisting of $N$ synchronous stochastic linear time-invariant (LTI) controlled processes exchanging information over a common resource-limited communication network with resource management capabilities (see Fig. \ref{fig:sys-model}). Each process $i\!\in \mathrm{N}\!\triangleq\!\{1,\ldots,N\}$ comprises of a physical plant $\mathcal{P}_i$, a delay-sensitivity controller $\mathcal{S}_i$, and a feedback control unit consisting of a state feedback controller $\mathcal{C}_i$ and an estimator $\mathcal{E}_i$. The dynamics of the plant $\mathcal{P}_i$, $i\in\mathrm{N}$, is described by the following stochastic difference equation:\vspace{-1mm}
\begin{equation}
x^i_{k+1}=A_ix^i_k+B_iu^i_k+w^i_k,
\label{eq:sys_model}\vspace{-1mm}
\end{equation}
where $x^i_k\!\in \!\mathbb{R}^{n^i}$ represents sub-system $i$'s state vector at time-step $k\!\in \!\mathbb{N}\cup \{0\}$, $u^i_k \!\in \!\mathbb{R}^{m^i}\!$ denotes the corresponding control signal, $w^i_k\!\in \!\mathbb{R}^{n^i}$ the stochastic exogenous disturbance, and $A_i\!\in \!\mathbb{R}^{n^i\times n^i}$ and $B_i\!\in \!\mathbb{R}^{n^i\times m^i}$ describe the system and input matrices, respectively. 
To allow for heterogeneity, $A_i$ and $B_i$ matrices can be different across the NCS, i.e., $A_i\neq A_j$ and $B_i\!\neq\! B_j$, $i,j\!\in \!\mathrm{N}$. The disturbances are assumed to be random sequences with independent and identically distributed (i.i.d.) realizations $w^i_k\!\sim \!\mathcal{N}(0,\Sigma_w^i)$, $\forall k$ and $i\!\in\!\mathrm{N}$, and $\Sigma_w^i\!\succ\!0$. The initial states~$x^i_0$'s are also presumed to be randomly selected from any arbitrary finite-moment distributions with variance $\Sigma_{x_0}^i$. For simplicity, we assume that the sensor measurements are perfectly noiseless copies of the state values\footnote{The results of this article extend, with lengthy but straightforward mathematical efforts, to noisy measurements if noise is an i.i.d. process.}.

\subsection{Communication system model}

To support the information exchange between each plant and its corresponding control unit, a resource-limited communication network exists that provides cost-prone latency-varying transmission services. 
More precisely, the communication network consists of a set of multiple distinct one-hop transmission links, represented by $\mathcal{L}\triangleq\{\ell_0,\ell_1,\ldots, \ell_D\}$, where $\ell_d$ denotes the transmission link with deterministic service latency of $d$ time-steps, and $|\mathcal{L}|\!=\!D\!+\!1$. Define the set $\mathcal{D}\!\triangleq\!\{0,1,\ldots,D\}$ and the vector $\Delta\!\triangleq\![0,1,\ldots,D]^\top$. Hence, if $x_k^i$ is sent to the controller $\mathcal{C}_i$ at time-step $k$ through the transmission link $\ell_d$ with $d$-step delay, $d\in\mathcal{D}$, then $x_k^i$ will be delivered to the controller at time-step $k+d$. Each transmission link $\ell_d\in \mathcal{L}$ is assigned a finite-valued service price $\lambda_d \in \mathbb{R}_{\geq 0}$ that is paid by the service recipient. Let $\Lambda\triangleq [\lambda_0,\lambda_1,\ldots,\lambda_D]^\top$  denote the prices assigned to the links in the transmission link set $\mathcal{L}$. The service prices are assigned such that shorter transmission delay induces higher price, i.e.,
$\lambda_0>\lambda_1 >\ldots>\lambda_D\geq 0$.

Denote $c_d\!\in \!\mathbb{N}$ as the transport capacity of a certain link $\ell_d\in \mathcal{L}$, which entails the link $\ell_d$ can transport at most $c_d$ number of data packets belonging to $c_d$ distinct sub-systems, simultaneously. The resource constraint can then be stated as
\begin{equation}\label{link-capacity-constraint}
c_d<N, \quad\forall \; d\in\mathcal{D}.
\end{equation}
Although, not all sub-systems can transmit through one certain link, we assume that the total capacity of all distinct transmission links is sufficient to service all sub-systems, via multiple transmission links, at every time-step $k\in\{0,1,\ldots\}$, i.e.,
\begin{equation}\label{tot-transmission}
\sum\nolimits_{d\in\mathcal{D}} c_d\geq N.
\end{equation}

\subsection{Distributed policy-makers \& decision variables}
We now introduce the policy makers and their corresponding decision outcomes for the underlying NCS, schematically depicted in Fig.~\ref{fig:sys-model}. The structural properties of the optimal policies will be thoroughly discussed in the next section.

\subsubsection{Delay-sensitivity} At the beginning of each sample cycle $k$ a local controller called ``\textit{delay controller}'' decides on delay-sensitivity of its corresponding sub-system by selecting one of the transmission links $\ell_d\!\in \!\mathcal{L}$. 
We define the binary-valued vector $\theta^i_k\triangleq [\theta^i_k(0),\ldots,\theta^i_k(D)]^{\textsf{T}}$ as the delay controller's decision variable of sub-system $i$ at time-step $k$, where each element of $\theta^i_k$ is determined as follows:
\begin{equation}\label{eq:delay-selector-var}
\theta^i_k(d)=\begin{cases} 1, &\!\!\!\!\!\!\!\!\!\text{link $\ell_d$ is selected to transmit $x_k^i$ at time $k$,} \\ 0, \qquad & \!\!\!\!\!\!\!\!\! \text{link $\ell_d$ is not selected.} \end{cases}
\end{equation}
We assume that each local delay controller selects only one of the transmission links per time-step, therefore, we have
\begin{equation}\label{const1}
\sum\nolimits_{d=0}^D\theta^i_k(d)=1, \quad \forall \; k\in\{0,1,\ldots\},\; \forall \; i\in\mathrm{N}.
\end{equation}

\subsubsection{Control input} The control unit of each local sub-system includes a feedback controller $\mathcal{C}_i$ and an estimator~$\mathcal{E}_i$, which are assumed collocated. At every time $k$, the control command $u_k^i\!\in\!\mathbb{R}^{m^i}$ is the outcome of a causal and measurable law 
 $\gamma_k^i(\cdot)$, given the available information at $\mathcal{C}_i$. In the absence of the state information $x_k^i$, the collocated estimator~$\mathcal{E}_i$ may calculate the state estimate $\hat{x}_k^i$ if it is required for the computation of $u_k^i$. 



\subsubsection{Resource allocation}
The constraint (\ref{link-capacity-constraint}) implies that if the number of requests to utilize a specific transmission link~$\ell_d$ exceeds the capacity $c_d$, not all requests can be accordingly serviced. 
Assume that a centralized network manager coordinates the resource allocation among sub-systems. In case $\sum_{i=1}^N \theta^i_k(d)\!> \!c_d$ for a certain link $\ell_d$, it decides which sub-systems will be serviced via the link $\ell_d$ and which~ones are reassigned to new transmission links. According to~(\ref{tot-transmission}), no scheduled data packet is dropped due to capacity limitation, as there will be another transmission link with free capacity to be assigned. We define the binary-valued vector $\vartheta_k^i\!\triangleq\! [\vartheta_k^i(0),\ldots,\vartheta_k^i(D)]^{\top}$ as the decision outcome of the centralized resource allocation mechanism that determines implementable transmission links for sub-system $i$. The element $\vartheta_{k}^i(d)\!\in\!\{0,1\}$ is similarly defined as in (\ref{eq:delay-selector-var}), except that it is determined by the network manager after receiving the requests from all the sub-systems. If at a time $k$, $\sum\nolimits_{i=1}^N \theta_{k}^i(d)\!\leq \!c_d, \forall d\!\in\!\mathcal{D}$, then $\vartheta_{k}^i \!=\!\theta_{k}^i, \forall i\!\in\!\mathrm{N}$. Otherwise, if $m$ requests are received for a certain link $\ell_d$ such that $m=\sum\nolimits_{i=1}^N \theta_{k}^i(d)> c_d$, new transmission links will be assigned to $m-c_d$ of those requests. 
This means for every sub-system $j$ of those~$c_d$ sub-systems, $\vartheta_{k}^j \!=\!\theta_{k}^j$ holds, while for every sub-system $\bar{j}$ belonging to the remaining set of $m- c_d$ sub-systems, $\vartheta_{k}^{\bar{j}} \!\neq\!d\in \mathcal{D} \setminus \{\tilde d,\bar d\}\theta_{k}^{\bar{j}}$. 
Element-wise, if a sub-system $\bar{j}$ requested a certain link $\ell_{\bar{d}}$, but instead was serviced with a different link $\ell_{\tilde{d}}$, then $\vartheta_{k}^{\bar{j}}(\tilde{d}) \!\neq\!\theta_{k}^{\bar{j}}(\tilde{d})$ and $\vartheta_{k}^{\bar{j}}(\bar{d})\neq\theta_{k}^{\bar{j}}(\bar{d})$, while for the rest of the elements, we have $\vartheta_{k}^{\bar{j}}(d) =\theta_{k}^{\bar{j}}(d)$, $\forall d\in \mathcal{D} \setminus\{\tilde{d},\bar{d}\}$. 

Since the ultimate link assignment is made by the network manager, state information received at the controller at time $k$, denoted by $\mathcal{Y}^i_k$, is determined by $\vartheta^i$. Define $y_{k-d}^i(d)\!=\!x_{k-d}^i$ if $\vartheta^i_{k-d}(d)\!=\!1$, and $y_{k-d}^i(d)\!=\!\emptyset$ if $\vartheta^i_{k-d}(d)\!=\!0$, then
\begin{equation}\label{set:received-state}
\mathcal{Y}^i_k=\{y_{k}^i(0),y_{k-1}^i(1),\ldots, y_{k-D}^i(D)\},
\end{equation}
where, to avoid notational inconvenience, we define $\vartheta_{-1}^i(d)\!=\!\vartheta_{-2}^i(d)\!=\!\ldots\!=\!\vartheta_{-D}^i(d)\!=\!0$ for all $d\in \mathcal{D}$. 

Out of order delivery is a common phenomenon that may happen depending on the selected resource allocation policy. Assume state $x^i_0$ is sent with delay 5 and $x^i_1$ is sent with zero delay, then $x^i_1$ will arrive before $x^i_0$.
However, out of delivery packet arrival will be adequately handled while constructing the state estimate and computing the control. 
If a stale state measurement arrives at the controller while a fresher one is available, the optimal controller uses only the \textit{freshest} one in constructing the optimal control input. 
Hence, without delving into the details, one can intuitively confirm that the optimal delay link profile should impose the least communication cost\footnote{Due to the constraint \eqref{const1} each sub-system is forced to pay a communication cost of at least $\lambda_D$ per time-step.} for \textit{outdated measurements}. This will naturally emerge as the solution of the optimization problems described later. 

\section{Problem Formulation: Joint Optimization}\label{sec:prob-state}

In this section, we formulate a cross-layer joint optimization problem and discuss its structural characteristics w.r.t. to the policy makers. 
The three decision makers are 1) local plant controllers that computes the control input $u_k^i,\;i\!\in\!\mathrm{N}$, at time-step $k$, 2) local delay controllers where the decision outcome $\theta_k^i$ determines the link~$\ell_d\!\in \!\mathcal{L}$ through which~$x_k^i$ will be transmitted, and 3) resource manager to compute $\vartheta_k^i$ that determines whether $\theta_k^i$ can be accordingly serviced. 

We assume that individual control systems have no knowledge of each other's parameters or decision variables. Let $\mathcal{I}_k^i$, $\mathcal{\bar{I}}_k^{i}$, and $\mathcal{\tilde{I}}_k$ denote the sets of accessible information for the plant controller, delay controller, and resource manager, respectively. (These sets are characterized in Section \ref{sec:optimal-co-design} where the information structure at each policy maker is discussed.). Then, at every time $k$, the plant control, delay control, and resource allocation policies are measurable functions of the $\sigma$-algebras generated by their corresponding information sets, i.e., $u_k^i \!=\! \gamma_k^i (\mathcal{I}_k^{i})$, $\theta_k^i \!=\! \xi_k^i (\mathcal{\bar{I}}_k^i)$, and $\vartheta_k\!=\!\pi_k(\mathcal{\tilde{I}}_k)$. Note that, $\gamma^i$ and $\xi^i$ represent local policies corresponding to a specific sub-system $i$, while $\pi$ is computed centrally and includes the resource allocation profile for all $i\!\in\!\mathrm{N}$. The local objective function of each sub-system $i\!\in\!\mathrm{N}$, denoted by~$J^i$, consists of its own LQG part plus the communication cost in average form over the finite horizon $[0,T]$, as follows:
\begin{align}\label{eq:local_objective}
\!\!\!J^i(u^i,\theta^i)\!=\!\E\!\Big[\|x_T^{i}\|_{Q_2^i}^2\!+\!\sum\nolimits_{k=0}^{T-1} \!\|x_k^{i}\|_{Q_1^i}^2\!+\!\|u_k^{i}\|_{R^i}^2\!+\!\theta_k^{i^\top}\!\!\Lambda\Big]\!
\end{align}
where, $Q_1^i\!\succeq \!0$, $Q_2^i\!\succeq \!0$, and $R^i\!\succ \!0$ represent constant weight matrices for the state and control inputs, respectively. 

The overall objective for the underlying NCS is to maximize the average performance of all sub-systems under the resource constraint (\ref{link-capacity-constraint}). This cannot simply be obtained by taking the average of the sum of the local cost functions (\ref{eq:local_objective}) because the local decision variable $\theta_k^i$ might not be realized due to the resource limitations. More precisely, the time that a state information is received at a controller might not always be the time decided by its delay controller. In fact, the cost function (\ref{eq:local_objective}) is achievable for a certain sub-system $i$ only if $\vartheta_k^i=\theta_k^i$, $\forall k\!\in\![0,T]$. However, if the capacity of one or more transmission links are exceeded by the number of requests, the resource manager adjusts some of those requests, which eventually changes the realization of the control signal $u_k^i$ and consequently the value of the local cost $J^i(u^i,\theta^i)$.

We formulate the system (commonly called social) cost $J$ as the average difference between the sum of $J^i$'s from the resource manager (given $\vartheta_k^i$'s) and local sub-systems' (given $\theta_k^i$'s) perspectives, i.e., knowing $\vartheta_k\!=\!\pi_k(\tilde{\mathcal{I}}_k\!)$, we have
\begin{equation}
J=\frac{1}{N}\sum\nolimits_{i=1}^N \E\!\Big[J^i(u^i,\vartheta^i)-\min_{u^i,\theta^i}J^i(u^i,\theta^i)\Big],
\label{eq:global_OP}
\end{equation}
and $J^i\!$ has been adjusted after resource allocation as 
\begin{align}\label{eq:local_objective_var}
\!\!\!J^i(u^i,&\vartheta^i)\!=\!\E\!\Big[\|x_T^{i}\|_{Q_2^i}^2\!+\!\sum\nolimits_{k=0}^{T-1} \!\|x_k^{i}\|_{Q_1^i}^2\!+\!\|u_k^{i}\|_{R^i}^2\!+\!\vartheta_k^{i^\top}\!\!\Lambda\Big]\!
\end{align} 
Note that, $J^i(u^i,\theta^i)$ is computed locally independent of the decisions for sub-systems $j\!\neq \!i$, while $J^i(u^i,\vartheta^i)$ is computed after central resource allocation is performed. 
The resources are allocated such that, w.r.t. the sub-systems preferences, the closest possible services are provided and $J$ is minimized. 

In addition to the delay controllers that determine the \textit{per-time} sensitivity of the control loops w.r.t. transmission latency, we introduce a constant latency-tolerance bound for each sub-system such that the resource manager allocates a transmission link only within that given bound. To diversify this static sensitivity for each sub-system, we define $\alpha_i$ and $\beta_i$ ($\in \!\mathcal{D}$) representing the maximum allowable delay tolerances. This specifies that a sub-system $i$ can tolerate imposed deviations by the network manager from the selected link $\ell_d$ only within the set $\{d-\alpha_i,\ldots,d,\ldots,d+\beta_i\}$\footnote{To avoid notational inconvenience, the network manager only takes into account the feasible tolerances of this set that also belong to $\mathcal{D}$. Moreover, for a nontrivial set, we assume at least one non-zero $\alpha_i$ and $\beta_j$, $i,j\in\mathrm{N}$.}. The ultimate goal is then finding the optimal policies $\gamma_k^{i,\ast} (\mathcal{I}_k^{i})$, $\xi_k^{i,\ast} (\mathcal{\bar{I}}_k^{i})$ and $\pi_k^\ast (\mathcal{\tilde{I}}_k)$ that jointly minimize the social cost $J$:
\begin{subequations}
\begin{align}\label{prob:global_OP}
&\!\!\min_{\gamma^i,\xi^i,\pi} J\\
\text{\small s. t.}\;\;\;\; & u_k^i = \gamma_k^i (\mathcal{I}_k^{i}),\;\; \theta_k^i = \xi_k^i (\mathcal{\bar{I}}_k^i),\;\;\vartheta_k=\pi_k(\mathcal{\tilde{I}}_k),\\\label{sensit_constrains}
& \!-\alpha_i\leq(\vartheta_k^i-\theta_k^i)^\top \Delta\leq \beta_i, \; i\in \mathrm{N},\\\label{prob:global_OP_last_line}
& \!\!\;  \sum\nolimits_{j=1}^N \vartheta_{k}^j(d)\leq c_d, \;d\in \mathcal{D}, \; k\!\in\![0,T-1].
\end{align}
\end{subequations}
Constraint (\ref{sensit_constrains}) specifies that if at  time $k$, $\theta_k^i(d)\!=\!1$, then the network manager allocates an available resource only from the set of links $\{\ell_{\max\{0,d-\alpha_i\}}, \ldots,\ell_{\min\{d+\beta_i,D\}}\}$ to sub-system $i$. 
The ultimate links from the allowable ones are selected by the resource manager such that the social cost $J$ is minimized. Note that problem (10) might not have a feasible solution for all $c_d$. We derive a sufficient feasibility condition in form of a lower bound for $c_d$ in the Section \ref{sec:optimal-co-design}.

Solving problem (10) is challenging due to the couplings between the decision variables. In fact, $\theta_k^i$ is the best choice, from the perspective of sub-system $i$, to make the balance between its LQG cost and communication price. However, delay controller decisions may go through changes because of resource limitations. Note that, the control input $u_k^i$ is explicitly affected by $\theta_k^i$ in the absence of the resource limitations, but if $\vartheta_k^i\!\neq \!\theta_k^i$, then $u_k^i$ will have a different~realization. This means the realization of $u_k^{i,\ast}$ computed~from the problem~(\ref{eq:local_objective}) might be different from that being computed from the problem~(\ref{eq:local_objective_var}) even if both are computed from the~same control law. Moreover, any decision of $\vartheta_k^i$ is clearly $\theta_k$-dependent. Further, as we discuss later, $\theta_{k+1}^i$ might also be a function of $\vartheta_{[0, k]}^i$. Altogether, problem (10) is nontrivial due to inter-dependencies and cross-layer constraints, hence we need to identify relevant conditions under which it can be decomposed. 

\section{Awareness Models \& Optimal Co-design}\label{sec:optimal-co-design}

Structural properties of the joint optimal policies are correlated with the cross-layer awareness model which characterizes the information sets $\mathcal{I}_k^i, \bar{\mathcal{I}}_k^i, \tilde{\mathcal{I}}_k$. We introduce different awareness models under which the couplings between $u_k^i$, $\theta_k^i$ and $\vartheta_k^i$ are examined. We discuss directed models of awareness for two different sets of information that can be exchanged between the decision makers of both layers: ``constant model parameters'' and ``dynamic variables''. In the rest of the article, awareness of the constant model parameters for the network layer, if assumed, entails the knowledge of $\{A_i, B_i, Q_1^i, Q_2^i, R^i, \Sigma_w^i,\Sigma_{x_0}^i\}, \forall i\!\in\!\mathrm{N}$. Note that, $\{\alpha_i,\beta_i\}$'s are known to the network layer. The local delay and plant controllers are also assumed to have the knowledge of their own model parameters $\{A_i, B_i, Q_1^i, Q_2^i, R^i, \Sigma_w^i,\Sigma_{x_0}^i,\alpha_i,\beta_i\}$ as well as the constant network parameters $\{\Lambda,\mathcal{L}\}$.

\begin{figure}[tb]
\centering
\psfrag{a}[c][c]{\scriptsize \text{Plant}}
\psfrag{aa}[c][c]{\scriptsize \text{controller}}
\psfrag{b}[c][c]{\scriptsize \text{Delay}}
\psfrag{bb}[c][c]{\scriptsize \text{controller}}
\psfrag{d}[c][c]{\scriptsize \text{Network}}
\psfrag{dd}[c][c]{\scriptsize \text{manager}}
\psfrag{ddd}[c][c]{\scriptsize \text{Application layer}}
\psfrag{c}[c][c]{\scriptsize \text{Network layer}}
\psfrag{g}[c][c]{\scriptsize $\vartheta_{k}^i$}
\psfrag{e}[c][c]{\scriptsize $u_{k}^i$}
\psfrag{f}[c][c]{\scriptsize $\theta_{k}^i$}
\psfrag{h}[c][c]{\tiny \text{Network model parameters}}
\psfrag{j}[c][c]{\tiny \text{System model parameters}}
\psfrag{t}[c][c]{\tiny \text{Here cycle $k$ begins}}
\includegraphics[width=6.5cm, height=4.5cm]{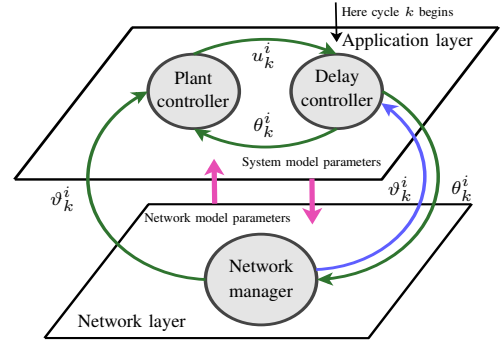}
\caption{Cross-layer interaction model: magenta arrows represent awareness of constant parameters. For network layer, awareness of system parameters, if assumed, includes $\{A_i, B_i, Q_1^i, Q_2^i, R^i, \Sigma_w^i,\Sigma_{x_0}^i,\alpha_i,\beta_i\}$, $\forall i\!\in\!\mathrm{N}$. For control loops, network parameters $\{\mathcal{L},\Lambda\}$ are known. If $\vartheta_k^i$ is available for the delay controller (violet arrow) we call the delay-control policy \textit{reactive}, otherwise, it is called \textit{impassive}.}
\label{fig:interaction_model1}
\vspace{-5mm}
\end{figure}

To discuss awareness of dynamic variables, it is essential to have a clear picture of the order of generating variables in one sample cycle, e.g., $k \!\rightarrow \!k\!+\!1$. At the beginning of a sample time $k$, the system state $x_k^i$ is updated according to the dynamics (\ref{eq:sys_model}), and then the delay controller generates $\theta_k^i$, based on the policy $\xi_k^i(\bar{\mathcal{I}}_k^i)$ to determine the transmission link through which $x_k^i$ is to be communicated. System state $x_k^i$ together with the service request $\theta_k^i$ is then forwarded to the network to be serviced. The resource manager receives this information from all sub-systems and checks whether the number of requests for each link is exceeding its capacity. It then computes $\vartheta_k^i$, according to the policy $\pi_k(\tilde{\mathcal{I}}_k)$, and $x_k^i$ is transmitted through the link determined by $\vartheta_k^i$. The control signal $u_k^i$ is computed from the control law $\gamma_k^i(\mathcal{I}_k^i)$\footnote{In case the information set $\mathcal{I}_k^i$ is not updated, i.e. if no new state information belonging to sub-system $i$ is scheduled to be delivered at time $k$, the control signal is updated based on a model-based estimation of $x_k^i$.}, $x_{k+1}^i$ is afterward updated and the pattern repeats over next samples. 

At the controllers, the following awareness model of the dynamic variables is valid throughout the article. Knowledge of the model parameters of sub-system $i$ is assumed for $\mathcal{C}_i$. Reminding (\ref{set:received-state}), the information set $\mathcal{I}_k^i$ at time $k$ is as
\begin{equation}\label{set:controller-information}
\mathcal{I}_k^i=\{\mathcal{Y}^i_0,... ,\mathcal{Y}^i_k,\theta_0^i,... ,\theta_k^i,\vartheta_0^i,... ,\vartheta_k^i,u_0^i,... ,u_{k-1}^i\}.
\end{equation}
As in Fig.~\ref{fig:interaction_model1}, the information set $\mathcal{I}_k^i$ in (\ref{set:controller-information}) specifies that the plant controllers are aware of the outcomes of the other two policies $\xi_{[0,k]}^i$ and $\pi_{[0,k]}^i$, from $t\!=\!0$ up to current time $t\!=\!k$. For that, we assume a dedicated low-bandwidth and error-free acknowledgement channel exists to inform the controllers at every time $k$ about $\theta_{k}^i$ and $\vartheta^i_{k}$ (see Fig. \ref{fig:sys-model}). 

To determine the awareness structure for the resource manager, we consider the following assumption: 

\textit{Assumption 1:} 
The resource allocation law $\pi_k$ is rendered independent of the local plant control policies $\gamma_{[0,k-1]}^i$, $i\!\in\!\mathrm{N}$.

Assumption 1 
declares a one-directional dependence between the plant control and resource allocation policies (see Fig.~\ref{fig:interaction_model1}), i.e., $\gamma_k^i$'s are explicit functions of $\vartheta_k^i$, but $\pi_k$ does not incorporate $u_{[0,k-1]}^i$'s, $i\!\in\!\mathrm{N}$, in determining $\vartheta_k^i$. Although this results in the resource allocation being independent of local control laws, $\pi_k$ depends on $\theta_{[0,k]}^i$ which itself is effected by the control signals. 
In other words, the local delay controllers generate $\theta_k^i$'s such that an averaged equilibrium is achieved between maximizing the control performance and minimizing the communication cost. Since $\pi_k$ is an explicit function of $\theta_{[0,k]}^i$'s, the effect of optimizing control performance is indirectly considered in resource allocation. hence, the explicit dependence between the plant control and the resource manager policies that requires full knowledge of $u_{[0,k-1]}^i$'s, $i\!\in\!\mathrm{N}$ at the resource manager, is avoided. This assumption, nonetheless, leads to a considerable complexity reduction in computing the optimal policies $\pi_k^\ast$ and $\gamma_k^{i,\ast}$ (Section \ref{subsec:CE}).

Having Assumption 1, we introduce the dynamic variables included in the resource manager's information set $\tilde{\mathcal{I}}_k$, as\vspace{-1mm}
\begin{equation}\label{set:RM-information}
\tilde{\mathcal{I}}_k=\{\theta_0,\ldots,\theta_k,\vartheta_0,\ldots,\vartheta_{k-1}\}.
\end{equation}
We also discuss the resource allocation with (Sec. \ref{sec:optimal-delay-model-aware}) and without (Sec. \ref{sec:no_model_awareness}) knowledge of the control systems model parameters. For the purpose of comparison, we discuss the scenario that the network manager does not take into account the local delay sensitivities in computing $\vartheta_k^i$'s, i.e., it allocates resources among sub-systems knowing neither the constant $\{\alpha_i,\beta_i\}$'s nor $\theta^i_{[0,k]}$'s, $\forall i\in\mathrm{N}$ (see Sec. \ref{sec:weighted_cost}). This is an important observation which shows how the local and social cost functions change w.r.t. the individual delay sensitivities.

For delay controllers, we introduce two design approaches, so called \textit{impassive} and \textit{reactive} delay control policies, each representing a distinct model of awareness of the dynamic variables (Fig.~\ref{fig:interaction_model1}). We derive the resulting joint optimal delay control and resource allocation policies in Sections \ref{sec:optimal-delay-model-aware} and \ref{sec:no_model_awareness}. Before that, to determine the structure of the optimal plant control policy~$\gamma_k^{i,\ast}$, $i\!\in\!\mathrm{N}$, 
we need to introduce the maximum amount of information that can be available at the $i^{\textsf{th}}$ delay controller at a time~$k$\footnote{Later we discuss that (\ref{set:DC-information}) corresponds to the reactive delay control approach and introduce the information set for the impassive approach.}. The set $\bar{\mathcal{I}}_k^i$ contains, at most, information about the following dynamic variables:\vspace{-1mm}
\begin{equation}\label{set:DC-information}
\bar{\mathcal{I}}_k^i=\{\theta_0^i,... ,\theta_{k-1}^i,\vartheta_0^i,... ,\vartheta_{k-1}^i,u_0^i,... ,u_{k-1}^i\}.
\end{equation} \vspace{-8mm}

\subsection{Certainty equivalence and optimal plant controller}\label{subsec:CE}

Having the sets $\mathcal{I}_k^i$, $\tilde{\mathcal{I}}_k$ and $\bar{\mathcal{I}}_k^i$ introduced in (\ref{set:controller-information})-(\ref{set:DC-information}), and reminding Assumption 1, we state the following theorem:

\begin{theorem}\label{thm:CE}
Given $\mathcal{I}_k^i$, $\tilde{\mathcal{I}}_k$ and $\bar{\mathcal{I}}_k^i$ in (\ref{set:controller-information})-(\ref{set:DC-information}) and under the Assumption 1, the optimal plant control law $\gamma_k^{i,\ast}$, $i\!\in\!\mathrm{N}$, w.r.t. (10) is of certainty equivalence form with the control inputs computed from the following linear state feedback law
\begin{align}\label{eq:CE-law}
u_k^{i,\ast}&=\gamma_k^{i,\ast}(\mathcal{I}_k^i)=-L_k^{i,\ast} \E[x_k^i|\mathcal{I}_k^i], \quad i\in\mathrm{N},\\\label{eq:CE-gain}
L_k^{i,\ast}&=\left(R^i+B_i^\top P_{k+1}^i B_i\right)^{-1}B_i^\top P_{k+1}^i A_i,
\end{align}
where, $P_T^i\!=\!Q_2^i$ and $P_k^i$ solves the below Riccati equation
\begin{align*}
P_k^i&\!=\!Q_1^i\!+\!A_i^\top \!\left[\!P_{\!k+1}^i\!-\!P_{\!k+1}^iB_i \!\left(R^i\!\!+\!B_i^\top \!P_{\!k+1}^i B_i\right)^{\!-1}\!B_i^\top \!P_{\!k+1}^i\!\right]\!A_i.
\end{align*}
\end{theorem}
\begin{proof}
See the Appendix \ref{Append:thm1}.
\end{proof}

\begin{remark}
In the absence of the constraint (\ref{link-capacity-constraint}), the resource allocation becomes redundant as $\vartheta_k^i\!=\!\theta_k^i$, $\forall i\!\in\!\mathrm{N}$ and $\forall k\!\in\![0,T]$. Hence, from (\ref{eq:global_cost_separation}), we have $\min_{\gamma^i,\xi^i,\pi} J=0$.
\end{remark}

\begin{corollary}\label{corl:optimal_value_function}
Under the optimal certainty equivalence control law (\ref{eq:CE-law})-(\ref{eq:CE-gain}), the optimal cost-to-go $V_k^{i,\ast}$ equals
\begin{align}
V_k^{i,\ast}&\!=\|\E\left[x_k^i|\mathcal{I}_k^i\right]\|^2_{P_k^i}\\\nonumber
&\!+\E\!\left[\|e_k^i\|^2_{P_k^i}+\!\sum\nolimits_{t=k}^{T-1}\!\|e_t^i\|^2_{\tilde{P}^i_t}\Big| \mathcal{I}_k^i\right]\!+\!\sum\nolimits_{t=k+1}^T \!\!\Tr (P_t^i \Sigma_w^i),
\end{align}
where, $e_k^i\triangleq x_k^i - \E\left[x_k^i|\mathcal{I}_k^i\right]$, and $\tilde{P}^i_t=Q_1^i+A_i^\top P_{t+1}^iA_i-P_t^i$. 
Moreover, the estimator, at time-step $k$, is given as follows
\begin{align}\label{eq:estimator_dynamics}
\!\!\!\!\!\!\E\!\left[x_k^i|\mathcal{I}_k^i\right]&\!=\!\sum\nolimits_{j=0}^{\text{min} \{D,k+1\}}\!b_{j,k}^i \E\!\left[x_k^i|x_{k-j}^i,u_0^i,...,u_{k-1}^i\right]\!,\!\!\!
\end{align}
and, for all $j\in\mathcal{D}$, and $k\geq j$, we have
\begin{align}\label{coeff:b}
b_{j,k}^i&=\prod\nolimits_{d=0}^{j-1}\prod\nolimits_{l=0}^d [1-\vartheta_{k-d}^i(l)][\sum\nolimits_{d=0}^j \vartheta_{k-j}^i(d)].
\end{align}
For, $k<j$, the $b_{0,k}^i,...,b_{k,k}^i$'s are defined as in (\ref{coeff:b}), $b_{k+1,k}^i\!=\!\prod_{d=0}^{k}\prod_{l=0}^d [1\!-\!\vartheta_{k-d}^i(l)]$, and for notational convenience, we define $b_{k+2,k}^i\!=\!...\!=\!b_{D,k}^i\!=\!0$.
\end{corollary}
\vspace{1mm}
\begin{proof}
The proof is similar to the proofs of Theorem 1 and Proposition 1 in \cite{8405590} and hence omitted for brevity.
\end{proof}

\begin{remark}
Theorem \ref{thm:CE} shows that the optimal control law is certainty equivalence (\ref{eq:CE-law}), yet $u_k^{i,\ast}$, i.e., the control law's realization, is computed based on $\E[x_k^i|\mathcal{I}_k^i]$ which is function on $\vartheta_{[k-D+1,k]}^i$, see \eqref{eq:estimator_dynamics}. 
We discuss in the next section that, if the delay controller is impassive, $V_k^{i,\ast}$ is estimated according to $\theta_{[0,k-1]}^i$. Thus, if at a time $t\!\in\![k\!-\!D,k\!-\!1]$, $\vartheta_t^i\neq \theta_t^i$, the delay controller computes $\E[V_k^{i,\ast}]$ as if $\theta_t^i$ is realized. Hence, $\E[V_k^{i,\ast}(\gamma^{i,\ast},\xi^i)]\!\neq \!\E[V_k^{i,\ast}(\gamma^{i,\ast},\pi)]$, despite similar $\gamma^{i,\ast}$ laws.
\end{remark}

\subsection{Optimal delay control and resource allocation policies}\label{sec:optimal-delay-model-aware}
We now derive optimal delay control and resource~allocation policies $(\xi_k^{i,\ast}\!,\pi_k^\ast)$ under the following two awareness models of the \textit{dynamic variables}. In this section, we assume the constant model parameters of all sub-systems are accessible for the network manager. Resource allocation without knowledge of constant parameters is studied in Section \ref{sec:no_model_awareness}.
\subsubsection{Impassive delay control}
We call the delay control policy an \textit{impassive process} if the decision on $\theta_k^i$'s is made independent of $\vartheta_{[0,k-1]}^i$, i.e., the delay controller is passive w.r.t. the resource manager's decisions. Hence, it decides on $\theta_k^i$'s knowing nothing about possible re-allocation by the resource manager. Therefore, the information set $\bar{\mathcal{I}}_k^i$ upon which $\theta_k^i=\xi_k^i(\bar{\mathcal{I}}_k^i)$ is computed impassively (see Fig. \ref{fig:info-topology}) becomes
\begin{equation}\label{set:DC-information-impassive}
\bar{\mathcal{I}}_k^i=\{\theta_0^i,... ,\theta_{k-1}^i,u_0^i,... ,u_{k-1}^i\}.
\end{equation}
Note that, although $\vartheta_{[0,k-1]}^i$ is not incorporated in computing $\theta_k^i$, the variable $\vartheta_k^i$ depends on $\{\theta_0, \ldots,\theta_k\}$. Moreover, the results of Theorem \ref{thm:CE} hold for $\bar{\mathcal{I}}_k^i$ in (\ref{set:DC-information-impassive}), as we have $\bar{\mathcal{I}}_k^i\!\subseteq\!\mathcal{I}_k^i$.

\begin{figure}[tb]
\centering
\psfrag{g}[c][c]{\tiny \text{LTI dynamics}}
\psfrag{a}[c][c]{\scriptsize $\bar{\mathcal{I}}_{k}^i$}
\psfrag{b}[c][c]{\scriptsize $\tilde{\mathcal{I}}_{k}$}
\psfrag{c}[c][c]{\scriptsize $\mathcal{I}_{k}^i$}
\psfrag{d}[c][c]{\scriptsize $\xi^i_k$}
\psfrag{f}[c][c]{\scriptsize $\pi_k$}
\psfrag{e}[c][c]{\scriptsize $\gamma_k^i$}
\psfrag{j}[c][c]{\scriptsize $u_{k}^i$}
\psfrag{k}[c][c]{\scriptsize $\vartheta_{k}^i$}
\psfrag{h}[c][c]{\scriptsize $x_{k+1}^i$}
\psfrag{m}[c][c]{\scriptsize $\theta_{k}^i$}
\psfrag{aa}[c][c]{\scriptsize $\bar{\mathcal{I}}_{k+1}^i$}
\psfrag{bb}[c][c]{\scriptsize $\tilde{\mathcal{I}}_{k+1}$}
\psfrag{cc}[c][c]{\scriptsize $\mathcal{I}_{k+1}^i$}
\psfrag{dd}[c][c]{\scriptsize $\xi^i_{k+1}$}
\psfrag{ff}[c][c]{\scriptsize $\pi_{k+1}$}
\psfrag{ee}[c][c]{\scriptsize $\gamma_{k+1}^i$}
\psfrag{jj}[c][c]{\scriptsize $u_{k+1}^i$}
\psfrag{kk}[c][c]{\scriptsize $\vartheta_{k+1}^i$}
\psfrag{hh}[c][c]{\scriptsize $x_{k+2}^i$}
\psfrag{mm}[c][c]{\scriptsize $\theta_{k+1}^i$}
\psfrag{CCC}[c][c]{\tiny \text{Centralized resource manager}}
\includegraphics[width=8cm, height=3.3cm]{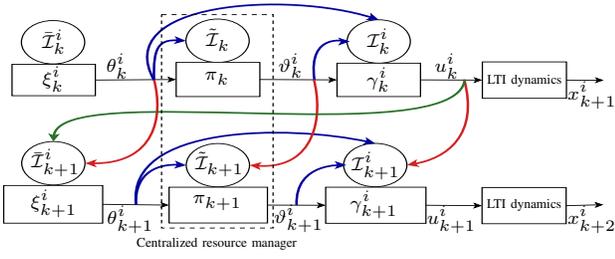}
\caption{Awareness model of the impassive delay control approach. Blue arrows represent policies' cross-awareness within one time-step. Red arrows show a policy maker's self-awareness. Green arrows depict state cross-awareness from one time-step to the next.}
\label{fig:info-topology}
\vspace{-5mm}
\end{figure}

\begin{theorem}\label{thm:impassiveDC}
Consider the problem (10) and let $\gamma^{i,\ast}, i\!\in \!\mathrm{N}$ follow the certainty equivalence law (\ref{eq:CE-law})-(\ref{eq:CE-gain}). Given $\bar{\mathcal{I}}_k^i$ and $\tilde{\mathcal{I}}_k$ in (\ref{set:DC-information-impassive}) and (\ref{set:RM-information}), the jointly optimal impassive delay control and resource allocation policies are offline solutions of the following constrained mixed-integer linear-programs (MILP)
\begin{align}\label{eq:opt-imp-del-cont}
&\theta_{[0,T-1]}^{i,\ast}=\argmin_{\xi_{[0,T-1]}^i}J^i(\gamma^{i,\ast},\xi_{[0,T-1]}^i(\bar{\mathcal{I}}_{[0,T-1]}^i))=\\\nonumber
&\argmin_{\xi_{[0,T-1]}^i} \sum_{t=0}^{T-1}\!\bigg[\theta_t^{i^\top}\!\Lambda\!+\!\sum_{l=0}^{\tau_{t}^i} \sum_{j=l}^{\tau_{t}^i}\bar{b}_{j,t}^i \textsf{Tr}(\tilde{P}_{t}^i A_{i}^{{l-1}^{\textsf{T}}} \Sigma_{w}^i A_{i}^{l-1})\bigg]\\\nonumber
&\text{s. t.}\quad \;\;\bar{b}_{j,t}^i=\prod\nolimits_{d=0}^{j-1}\prod\nolimits_{l=0}^d [1-\theta_{t-d}^i(l)][\sum\nolimits_{d=0}^j \theta_{t-j}^i(d)],\\\nonumber
& \qquad\;\;\;\sum\nolimits_{l=0}^D \theta_t^i(l)\!=\!1,
\quad\!\!\!\!\sum\nolimits_{j=0}^{\tau_t^i} \bar{b}_{j,t}^i\!=\!1, \quad \!\!\!\!\sum\nolimits_{j=t+2}^{D} \bar{b}_{j,t}^i\!=\!0,
\end{align}
and, \vspace{-3mm}
\begin{align}\label{eq:opt-imp-res-alloc}
&\!\!\!\vartheta_{[0,T-1]}^{\ast}\!=\!\argmin_{\pi_{[0,T-1]}}\!\frac{1}{N}\!\sum_{i=1}^N J^{i}(\gamma^{i,\ast}\!,\pi_{[0,T-1]}(\tilde{\mathcal{I}}_{[0,T-1]}))\!=\!\\\nonumber
&\!\!\!\argmin_{\pi_{[0,T-1]}}\frac{1}{N}\!\sum_{i=1}^N \sum_{t=0}^{T-1}\!\left[\vartheta_t^{i^\top} \!\!\Lambda\!+\!\sum_{l=0}^{\tau_{t}^i} \sum_{j=l}^{\tau_{t}^i}b_{j,t}^i \textsf{Tr}(\tilde{P}_{t}^i A_{i}^{{l-1}^{\textsf{T}}} \Sigma_{w}^i A_{i}^{l-1})\!\right]\\\nonumber
&\!\!\!\text{s. t.} \quad\; -\alpha_i\leq(\vartheta_t^i-\theta_t^{i,\ast})^\top \Delta\leq \beta_i, \; b_{j,t}^i \;\text{as defined in (\ref{coeff:b})},\\\nonumber
&\!\!\!\qquad \quad\sum\nolimits_{i=1}^N \vartheta_{t}^i(d)\leq c_d,\; \forall d\in\mathcal{D}, \; t\in[0,T-1].
\end{align}
where, $\tau_t^i\!\triangleq \!\min\{D,t+1\}$. 
\end{theorem}
\begin{proof}
See the Appendix \ref{Append:thm2}.
\end{proof}

Next, we propose (without a proof) a sufficient capacity condition for $c_d, d\!\in \!\mathcal{D}$, ensuring the allocated resources~are within $\{\ell_{\max\{0,d-\alpha_i\}}, \ldots, \ell_d,\ldots,\ell_{\min\{d+\beta_i,D\}}\}$, and the MILP (\ref{eq:opt-imp-res-alloc}) is feasible. Selected $c_d$'s should additionally satisfy (\ref{link-capacity-constraint}) and (\ref{tot-transmission}) to ensure the problem (10) is non-trivial, and avoid packet drop. We show in Section~\ref{num_res} that the condition is not necessary.
\begin{corollary}\label{corol:feasibility}
The MILP problem (\ref{eq:opt-imp-res-alloc}) is feasible if (\ref{tot-transmission}) is satisfied and $\forall d\!\in \!\mathcal{D}$, the following sufficient condition holds
\begin{equation}\label{feasibility}
c_d\geq \bigg\lfloor\!\frac{N}{1\!+\!\frac{1}{N}\left[h(\alpha,\beta)\right]}\!\bigg\rfloor,
\end{equation} 
with, $h(\alpha,\beta)\!=\!\sum_{i\in N_1} \mathbbm{1}(d\alpha_i)\!+\!\sum_{j\in N_2} \mathbbm{1}((D-d)\beta_j)\!+\!\mathbbm{1}(d)\sum_{l\in N_3} \mathbbm{1}(d\alpha_l) \!+\!\mathbbm{1}(D-d)\sum_{l\in N_3} \mathbbm{1}((D-d)\beta_l)$, where, $\forall i\!\in\! N_1$, $j\!\in\! N_2$ and $l\!\in \!N_3$, we have $(\alpha_i\!\neq\!0, \beta_i\!=\!0)$, $(\alpha_j\!=\!0, \beta_j\!\neq\!0)$, $(\alpha_l,\beta_l\!\neq\!0)$ and $|N_1|\cup |N_2|\cup |N_3|=N$.
\end{corollary}

\subsubsection{Reactive delay control}
We call the delay control policy \textit{reactive} if the decisions on $\theta_k^i$'s are per-time made incorporating the knowledge of $\vartheta_{[0,k-1]}^i$. Thus, the information set $\bar{\mathcal{I}}_k^i$ upon which $\theta_k^i=\xi_k^i(\bar{\mathcal{I}}_k^i)$ is computed, needs to contain $\vartheta_{[0,k-1]}^i$, hence $\bar{\mathcal{I}}_k^i$ coincides with (\ref{set:DC-information}).

\begin{theorem}\label{thm:reactiveDC}
Consider the optimization problem (10). Let $\gamma^{i,\ast}, i\!\in \!\mathrm{N}$ follow the certainty equivalence law (\ref{eq:CE-law})-(\ref{eq:CE-gain}). Given the information sets $\bar{\mathcal{I}}_k^i$ and $\tilde{\mathcal{I}}_k$, respectively, in (\ref{set:DC-information}) and (\ref{set:RM-information}), the optimal reactive delay control law is computed online from the following constrained MILP
\begin{align}\label{eq:thm3-del-ctrl}
&\theta_{[k,T-1]}^{i,\ast}\!=\argmin_{\xi_{[k,T-1]}^i}J^i(\gamma^{i,\ast},\xi_{[k,T-1]}^i(\bar{\mathcal{I}}_{[k,T-1]}^i))=\\\nonumber
&\argmin_{\xi_{[k,T-1]}^i} \sum_{t=k}^{T-1}\!\left[\theta_t^{i^\top}\!\Lambda\!+\!\sum_{l=0}^{\tau_{t}^i} \sum_{j=l}^{\tau_{t}^i}\tilde{b}_{j,t}^i \textsf{Tr}(\tilde{P}_{t}^i A_{i}^{{l-1}^{\textsf{T}}} \Sigma_{w}^i A_{i}^{l-1})\right]\\\nonumber
&\text{s. t.}\quad \!\!\tilde{b}_{0,t}^i=\theta_t^i(0), \quad \tilde{b}_{j,t}^i \leq \sum\nolimits_{l=0}^j \vartheta_{t-j}^i(l),\; j\!\in\!\{1,\ldots,\tau_t^i\},\\\nonumber
&\qquad\sum\nolimits_{l=0}^D \!\theta_t^i(l)\!=\!1, \;\sum\nolimits_{j=0}^{\tau_t^i} \!\tilde{b}_{j,t}^i\!=\!1, \; \sum\nolimits_{j=t+2}^{D} \!\tilde{b}_{j,t}^i\!=\!0, t\!\geq \!k.
\end{align}
where, $\tau_t^i$ and $\tilde{P}_{t}^i$ are similarly defined as in Theorem \ref{thm:impassiveDC}, and
\begin{equation*}
\tilde{b}_{j,t}^i\!=\!\!\Big[[1\!-\!\theta_t^i(0)]\prod\nolimits_{d=1}^{j-1}\prod\nolimits_{l=0}^d [1\!-\!\vartheta_{t-d}^i(l)]\Big]\!\Big[\!\sum\nolimits_{d=0}^j \!\vartheta_{t-j}^i(d)\Big]\!,
\end{equation*}
with $\prod_{d=1}^{0}\prod_{l=0}^d [1-\vartheta_{t-d}^i(l)]\triangleq 1$, for notation convenience. 

Moreover, the optimal resource allocation law is computed online from the following constrained MILP
\begin{align}\label{eq:thm3-res-manager}
&\vartheta_{[k,T-1]}^{\ast}=
\argmin_{\pi_{[k,T-1]}}\frac{1}{N}\!\sum\nolimits_{i=1}^N \sum\nolimits_{t=k}^{T-1}\bigg[\vartheta_t^{i^\top} \Lambda\\\nonumber
&\quad\qquad\;\;+\sum\nolimits_{l=0}^{\tau_{t}^i} \sum\nolimits_{j=l}^{\tau_{t}^i} b_{j,t}^i \textsf{Tr}(\tilde{P}_{t}^i A_{i}^{{l-1}^{\textsf{T}}} \Sigma_{w}^i A_{i}^{l-1})\bigg]\\\nonumber
&\text{s. t.} \quad\; -\alpha_i\leq(\vartheta_t^i-\theta_t^{i,\ast})^\top \Delta\leq \beta_i, \; b_{j,t}^i \;\text{as defined in (\ref{coeff:b})},\\\nonumber
&\qquad \quad\sum\nolimits_{i=1}^N \vartheta_{t}^i(d)\leq c_d,\; \forall d\in\mathcal{D}, \; t\in[k,T-1].
\end{align}
\end{theorem}

\vspace{1mm}
\begin{proof}
Derivation of optimal policies in Theorem \ref{thm:reactiveDC} follows similarly to that of Theorem \ref{thm:impassiveDC} and hence omitted. The major differences are summarized in the Remark \ref{remark:3}.
\end{proof}
\begin{remark}\label{remark:3}
In Theorem \ref{thm:reactiveDC}, the reactive delay controller is aware of $\vartheta_{[0,k-1]}^{i,\ast}$ and incorporates them in deciding~$\theta_{[k,T-1]}^{i,\ast}$. Hence, unlike Theorem \ref{thm:impassiveDC}, here we solve a per-time-step MILP. Technically, the online nature of the MILP (\ref{eq:thm3-del-ctrl}) is reflected in the time-varying $\tilde{b}_{j,t}^i$ that results in a time-varying $\theta_{[k,T-1]}^{i,\ast}$. Comparing it with $\bar{b}_{j,t}^i$ in Theorem \ref{thm:impassiveDC}, we see that for each time $k$, $\theta_{[k,T-1]}^{i,\ast}$ depends on $\vartheta_{[k-D,k-1]}^{i,\ast}$, while in Theorem \ref{thm:impassiveDC} the same decision was dependent only on $\theta_{[k-D,k-1]}^{i,\ast}$. The MILP problem (\ref{eq:thm3-res-manager}) also becomes online as it needs to satisfy the time-varying constraint $-\alpha_i\leq(\vartheta_t^i-\theta_t^{i,\ast})^\top \Delta\leq \beta_i$. 
\end{remark}

\begin{remark}\label{rem:complexity}
The optimal impassive delay control and resource allocation variables $(\theta_{[0,T-1]}^\ast,\vartheta_{[0,T-1]}^\ast)$ are shown in Theorem~\ref{thm:impassiveDC} to be offline solutions of the MILPs (\ref{eq:opt-imp-del-cont}) and (\ref{eq:opt-imp-res-alloc}), while the same variables of the reactive approach are solved online from the MILPs (\ref{eq:thm3-del-ctrl}) and (\ref{eq:thm3-res-manager}), as in Theorem~\ref{thm:reactiveDC}. 
Based on their formulations, the impassive approach requires an MILP of complexity $\mathcal{O}(NdT)$ whereas the reactive approach requires an MILP of complexity $\mathcal{O}(NdT^2)$. This confirms that both approaches incur linear complexity growth w.r.t. to the number of sub-systems and the number of transmission links. However, complexity of the reactive approach grows quadratically with the time horizon length while the respective growth rate for the impassive approach is linear\footnote{A less complex scenario can be discussed when the control systems decide on a desired transmission link not for a single time-step but for a time interval. The joint optimal solution for such a scenario can be derived similar to the results of this article yet the computational complexity is~reduced. The social cost, however, will be higher as constraints are per interval.}.
\end{remark}

\begin{remark}
According to (\ref{eq:estimator_dynamics}), the state estimation at the controller is performed using the freshest received state information, hence, if an outdated state arrives while a fresher one is available, the former will not be used. In addition, both local and social objective functions (\ref{eq:local_objective})-(\ref{eq:global_OP}) include communication costs. Therefore, to reduce the total cost, the delay controllers and the resource manager try to avoid transmission decisions that lead to out of order delivery of state information. This is reflected in the formulated MILPs in Theorems \ref{thm:impassiveDC} and~\ref{thm:reactiveDC}. This is, however, unavoidable due to the constraint (\ref{const1}) that forces each sub-system to select one delay link $\ell_d\!\in\!\mathcal{L}$ while the maximum delay $D$ is finite. Intuitively, many of transmissions with $D$-step delay would not have been executed if the sub-systems had the option to remain open-loop and select \textit{no transmission}. Hence, outdated information appearing at subsequent time-steps are discarded if a fresher data exists.

\end{remark}

Corollary \ref{corol:performance-comparison} below shows that, although the reactive approach requires more computation, it outperforms the impassive approach in terms of both local and social performances.

\begin{corollary}\label{corol:performance-comparison}
Let the performance of the local policy co-design $(\gamma^{i,\ast},\xi^{i,\ast},\pi^\ast)$ for the impassive and reactive approaches be denoted, respectively, by $J^{i,\ast}_{\text{Im}}$ and $J^{i,\ast}_{\text{Re}}$, defined in (\ref{eq:local_objective}), and also denote the social performance of the overall joint design $(\gamma^{\ast},\xi^{\ast},\pi^\ast)$ by $J^\ast_{\text{Im}}$ and $J^\ast_{\text{Re}}$, defined in (\ref{eq:global_OP}). Let $\gamma^{i,\ast}$, $\xi^{i,\ast}$ and $\pi^\ast$ of the impassive approach be computed as (\ref{eq:CE-law}), (\ref{eq:opt-imp-del-cont}) and (\ref{eq:opt-imp-res-alloc}), and of the reactive approach as (\ref{eq:CE-law}), (\ref{eq:thm3-del-ctrl}) and (\ref{eq:thm3-res-manager}), respectively. Then, $J^{i,\ast}_{\text{Re}}\leq J^{i,\ast}_{\text{Im}}$ and $J^\ast_{\text{Re}}\leq J^\ast_{\text{Im}}$.
\end{corollary}
\begin{proof}
See the Appendix \ref{Append:corol3}.
\end{proof}

\subsection{Optimal resource allocation without model awareness} \label{sec:no_model_awareness}
 In an NCS, the individual entities may not be willing to share the specifications of their dynamical model or their objective functions with the communication service provider. Within our problem formulation, this essentially means that the network manager does not have the knowledge of constant parameters $\{A_i, B_i, Q_1^i, Q_2^i, R^i, \Sigma_w^i,\Sigma_{x_0}^i\}$, $i\in \mathrm{N}$. Technically, having no knowledge of the constant parameters (except $\alpha_i,\beta_i$) the local cost functions $J^i$ are not computable for the network manager, hence the optimal resource allocation policy cannot be obtained from the problem (\ref{prob:global_OP}). More precisely, although the local policies $\gamma^{i,\ast}$'s and $\xi^{i,\ast}$'s can still be computed from (\ref{eq:CE-law}), (\ref{eq:opt-imp-del-cont}), and (\ref{eq:thm3-del-ctrl}), for impassive and reactive approaches, respectively, $\pi^\ast$ cannot be obtained from the either problems (\ref{eq:opt-imp-res-alloc}) and (\ref{eq:thm3-res-manager}). Let the information set $\tilde{\mathcal{I}}_k$ the network manager be defined as in (\ref{set:RM-information}) but excluding the knowledge of the constant parameters of all sub-systems except $\alpha_i,\beta_i$'s. Then the best the network manager can perform is to allocate resources such that, given $\alpha_i, \beta_i$'s, the average deviation between the delay control and resource allocation decisions is minimized, which is the first term in the MILPs (\ref{eq:opt-imp-res-alloc}) and (\ref{eq:thm3-res-manager}).  
Hence, the optimal resource allocation for the impassive approach will be obtained from
\begin{align}\label{prob:res_alloc_no_knowledge_imp}
&\vartheta_{[0,T-1]}^{\ast}=\argmin_{\pi_{[0,T-1]}}\frac{1}{N}\sum\nolimits_{i=1}^N \sum\nolimits_{t=0}^{T-1}\vartheta_t^{i^\top} \!\Lambda\\\nonumber
&\text{s. t.} \quad\; -\alpha_i\leq(\vartheta_t^i-\theta_t^{i,\ast})^\top \Delta\leq \beta_i, \; i\in\mathrm{N},\\\nonumber
&\qquad \quad\sum\nolimits_{i=1}^N \vartheta_{t}^i(d)\leq c_d,\; \forall d\in\mathcal{D}, \; t\in[0,T-1],
\end{align}
and for the reactive approach, is obtained from
\begin{align}\label{prob:res_alloc_no_knowledge_re}
&\vartheta_{[k,T-1]}^{\ast}=\argmin_{\pi_{[k,T-1]}}\frac{1}{N}\sum\nolimits_{i=1}^N \sum\nolimits_{t=k}^{T-1}\vartheta_t^{i^\top}\Lambda\\\nonumber
&\text{s. t.} \quad\; -\alpha_i\leq(\vartheta_t^i-\theta_t^{i,\ast})^\top \Delta\leq \beta_i,\; i\in\mathrm{N},\\\nonumber
&\qquad \quad \sum\nolimits_{i=1}^N \vartheta_t^i(d)\leq c_d, \; \forall d\in\mathcal{D}, \; t\in [k,T-1],
\end{align}
where, $\theta_t^{i,\ast}$ in (\ref{prob:res_alloc_no_knowledge_imp}) is the solution of the impassive approach (\ref{eq:opt-imp-del-cont}), while in (\ref{prob:res_alloc_no_knowledge_re}) is solution of the reactive approach (\ref{eq:thm3-del-ctrl}).

From (\ref{prob:res_alloc_no_knowledge_imp}) and (\ref{prob:res_alloc_no_knowledge_re}), in the absence of the constant model parameters the resource manager only optimizes the communication cost, and that the allocated resource to remain within the sensitivity constraint (\ref{sensit_constrains}). This results in a solution for $\vartheta$ that tends to select the transmission links that incur the least communication cost ignoring that such selections may severely affect the control cost.
To counter that, in the reactive approach where the delay controller can adjust its link selection profile in response to the resource allocation policy, each system changes their $\theta^{i,\ast}_k$ drastically for the future time-steps to request for faster links aiming to reduce the control cost. 
Assume a system asked for a fast link, e.g. with delay zero, due to its task criticality, however, the network manager does not realize the urgency due to not being capable of estimating the control cost and allocates a higher latency transmission link (say $d\!=\!2$) which optimizes only the communication cost. The system will then be forced to select a low delay link again since its past request is not served accordingly. This approach thus leads to higher total cost of control and communication compared to the scenario that the resource manager knows the constant model parameters. 
Furthermore, when constant model parameters are assumed unknown, the reactive approach performs significantly better than its impassive counterpart since the systems will be generally unhappy of this agnostic resource allocation, hence respond with a significantly different $\theta_k^\ast$ than the prescribed $\vartheta_k^\ast$ that leads to a very different $\vartheta^\ast_{k+1}$ than $\vartheta^\ast_k$.

\subsection{Delay-insensitive optimal resource allocation}\label{sec:weighted_cost}
For the purpose of benchmarking and comparing the two methods presented in the previous sections, we propose another ad-hoc approach by extending the work of \cite{8405590} to a multi-agent scenario.
More specifically, the approach presented in this section adopts a formulation that does not consider the delay sensitivity in the formulation, rather solely interested in the capacity constraint.
This means that the resource manager ignores the knowledge of $\theta_{[0,k]}^i$ and $\{\alpha_i,\beta_i\}$'s, $i\!\in\!\mathrm{N}$, however, knows the constant model parameters of all sub-systems.
We define constant weights $w_i \!>\!0$ such that $\sum_{i=1}^N w_i \!=\!1$. The network manager then prioritizes each sub-system based on $w_i$ and optimizes the MILP at every time-step $k$, i.e.,
\begin{align} \label{prob:res_alloc_weighted_re}
&\vartheta_{[k,T-1]}^{\ast}\!=\argmin_{\pi_{[k,T-1]}}\sum_{i=1}^N w_i\E\!\bigg[\!V^{i,\ast}_k(\gamma^{i,\ast}\!,\pi^i)\!+\!\!\sum_{t=k}^{T-1}\vartheta_t^{i^\top}\!\Lambda\big|\tilde{\mathcal{I}}_k\bigg]\!=\nonumber \\\nonumber
&\argmin_{\pi_{[k,T-1]}}\!\sum_{i=1}^N\sum_{t=k}^{T-1}\!w_i\bigg[\vartheta_t^{i^\top} \!\!\Lambda\!+\!\sum_{l=0}^{\tau_{t}^i} \sum_{j=l}^{\tau_{t}^i} b_{j,t}^i \textsf{Tr}(\tilde{P}_{t}^i A_{i}^{{l-1}^{\textsf{T}}} \Sigma_{w}^i A_{i}^{l-1})\bigg]\!\\
&\text{s. t.} 
\quad \sum\nolimits_{i=1}^N \vartheta_{t}^i(d)\leq c_d,\; \forall d\in\mathcal{D}, \; t\in[k,T-1].
\end{align}

Notice that since there is no coupling between $\vartheta_t$ and $\theta_t$ contrasting to the formulations in \eqref{eq:thm3-res-manager} and \eqref{prob:res_alloc_no_knowledge_re}, $\vartheta^{\ast}_{[k,T-1]}$ can be found from $\vartheta^{\ast}_{[0,T-1]}$ without solving \eqref{prob:res_alloc_weighted_re} for all $k$. 
In fact if $\vartheta^{\ast}_{[0,T-1]}$ is the solution of \eqref{prob:res_alloc_weighted_re} for $k=0$, then the part $\vartheta^{\ast}_{[t,T-1]}$ of $\vartheta^{\ast}_{[0,T-1]}$ is the solution of \eqref{prob:res_alloc_weighted_re} for any $k=t$.
Furthermore, any feasible solution of \eqref{eq:thm3-res-manager} is a feasible solution for \eqref{prob:res_alloc_weighted_re}, and hence, often the delay-insensitive approach results in a lower social cost than the delay-sensitive MILP in \eqref{eq:thm3-res-manager}.
However, the lower social cost in this approach is obtained at the expense of higher deviations between the desired links and the allocated ones since no constraint of the form $-\alpha_i\leq(\vartheta_t^i-\theta_t^{i,\ast})^\top \Delta\leq \beta_i$ exists to restrict the deviation between $\vartheta_t^i$ and $\theta_t^{i,\ast}$. 
Hence, the social performance is expected to improve, however, certain individual sub-systems suffer as their link allocation is far from the ones requested. This trade-off needs to be attended for the resource manager to be sufficiently responsive to timeliness sensitivity of local sub-systems.

%% file: numerics.tex
\section{SIMULATION RESULTS}\label{num_res}
We consider an NCS consisting of 10 homogeneous stable and 10 homogeneous unstable sub-systems. The system and input matrices for the unstable and stable groups are $A^u\!=\!\begin{bmatrix} 1.01 & \!\!0.2\\ 0.2 & \!\!1 \end{bmatrix}$, $A^s\!=\!\begin{bmatrix} 0.5 & \!\!0.1\\ 0.6 & \!\!0.8 \end{bmatrix}$, and $B^u\!=\!B^s\!=\!\begin{bmatrix} 0.1 &\!\!0 \\ 0 &\!\! 0.15 \end{bmatrix}$, respectively. The disturbance is Gaussian distributed with mean and variance as $\mathcal{N}(0,1.5I_2)$. The LQG cost parameters for all sub-systems are identically set as $Q^i_1\!=\!Q^i_2\!=\!R^i\!=\!I_2$, and $T\!=\!20$ is the total time horizon of the simulations. 

The network supports the control loops via~$6$ transmission links with delays of $d\!\in\![0,1,2,3,4,5]$ time-steps associated with the cost $\Lambda\!=\![25,17,11,7,4,1]$.
We assume $c_d\!=\!6$, $\forall d$, and $\alpha_i\!=\!\beta_i\!=\!3$, $\forall i\!\in\!\{1,\ldots, 20\}$. Note that $c_d\!=\!6$ satisfies the individual and total capacity constraints (\ref{link-capacity-constraint})~and~(\ref{tot-transmission}), however, does not meet the sufficient feasibility condition~(\ref{feasibility}) for $d\!=\!\{0,5\}$\footnote{According to (\ref{feasibility}), $c_d\!\geq \!10$ for $d\!=\!\{0,5\}$ and $c_d\!\geq \!6$ for $d\!=\!\{1,2,3,4\}$.} and yet is a valid choice for this simulation setup, which shows (\ref{feasibility}) is not a necessary condition.
\begin{figure}
\centering
\includegraphics[trim=10 10 20 20, clip,width=.91 \linewidth]{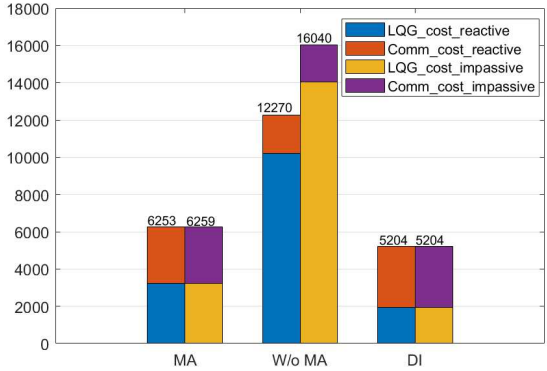}\vspace{-2mm}
\caption{Optimal costs for different approaches. MA: with model awareness (Sec.~\ref{sec:optimal-delay-model-aware}), W/o MA: without model awareness (Sec.~\ref{sec:no_model_awareness}), DI: delay-insensitive (Sec.~\ref{sec:weighted_cost}).} \label{fig:cost_comparison}
\vspace*{-4mm}
\end{figure}

We illustrate the optimal delay control and link allocation for each sub-system using the discussed approaches: 1) with model awareness, 2) without model awareness, and 3) delay-insensitive approach, as presented in sections \ref{sec:optimal-delay-model-aware}, \ref{sec:no_model_awareness}, and \ref{sec:weighted_cost}, respectively.
For the first two approaches, we employ both \textit{reactive} and \textit{impassive} methods to perform optimal co-design and compare their outcomes.
As discussed in Corollary \ref{corol:performance-comparison}, we demonstrate that the reactive method performs no worse than the impassive method and may often perform significantly better, due to the dynamic coupling between $\theta$ and $\vartheta$. Since such coupling does not exist in the delay-insensitive case, reactive and impassive methods yield identical results.

In Fig. \ref{fig:cost_comparison}, we illustrate the LQG control and communication costs for the above-mentioned approaches, where we see that the awareness of the constant model parameters leads to a significant performance improvement when compared with no model awareness scheme.
However, as also discussed in Section \ref{sec:no_model_awareness}, the superiority of the reactive approach over the impassive counterpart is far better for the case without model awareness. 
In fact, one needs to contemplate whether to employ the reactive approach when the network manager has access to the constant model parameters, due to the insignificant overall performance augmentation at the expense of the extra computational complexity (see Remark~\ref{rem:complexity}).

%

Fig.~\ref{fig:utimulti_inhomo} shows the transmission link utilization profile (defined in \ref{eq:link_utilization_numerics}) where we only provide the plot for the impassive and reactive scenarios when the network manager is not aware of the constant model parameters (Section \ref{sec:no_model_awareness}).
\begin{align}\label{eq:link_utilization_numerics}
    \rho_i(t) =\frac{\# \text{ of utilization of Link } i \text{ up to time }t}{N(t+1)}.
\end{align}
\begin{figure}[t]
    \centering
    \includegraphics[trim=10 20 15 25, clip, width=1.00 \linewidth]{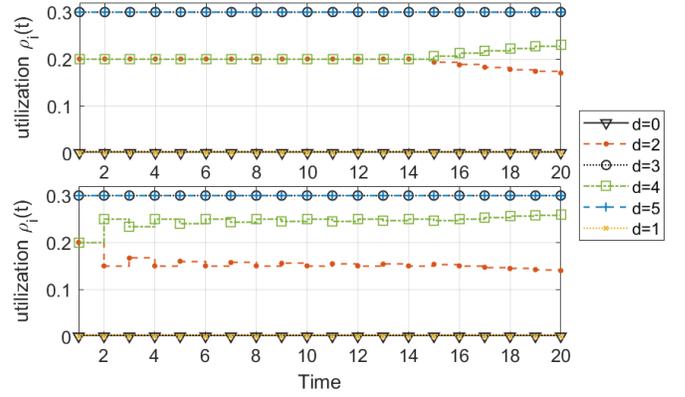}
    \vspace{-5mm}\caption{Link utilization over time under capacity constraints without model awareness based approach.
    Top: reactive method, bottom: impassive method.}
    \label{fig:utimulti_inhomo}
        \vspace{-4mm}
\end{figure}
According to \eqref{eq:link_utilization_numerics}, $\sum_{i=1}^N \rho_i(t)=1$ at every time $t$, that is also reflected in Fig.~\ref{fig:utimulti_inhomo}. For the case without model awareness, the network manager only cares about the communication cost and hence the cheaper links are utilized, as can be seen in Fig. \ref{fig:utimulti_inhomo}. 
Notice that link 3 is used more than link 4 due to the coupling constraints between $\theta_t$ and $\vartheta_t$ in \eqref{prob:res_alloc_no_knowledge_imp} and \eqref{prob:res_alloc_no_knowledge_re}. 
The sub-systems which requested for the link $\ell_0$, can not be assigned to any link beyond $\ell_3$ since $\beta_i\!=\!3$. 
Thus, the majority of the requests for link $\ell_0$ were assigned to $\ell_3$ and the rest were assigned to $\ell_2$ ($\ell_1$ is more expensive).
Similarly, the majority of the requests for $\ell_5$ are assigned to $\ell_5$ and the rest to $\ell_4$, etc.

We also studied this problem for the case with model awareness, and we noticed that the difference in the link utilization is minor between the two impassive and reactive approaches (as also corroborated by the cost difference in Fig. \ref{fig:cost_comparison}). 
In fact, the link utilization, in this case, changes only after time $t\!=\!15$.
This observation brings out the question whether it makes sense to adopt the computationally expensive reactive approach over the simple impassive approach for this little improvement. 
Based on this observation, one may be tempted to adopt reactive approach in an intermittent fashion, i.e., instead of solving \eqref{eq:thm3-res-manager} for every $k$, do so at $k\!=\!t_1,t_2,\ldots,t_\ell$ where $0\!<\! t_1 \!<\!\ldots\!<\!t_\ell\!<\!T$.
An interesting yet challenging research question is how to determine $t_1,\ldots,t_\ell$. One may perhaps adopt an event-based strategy to solve for these quantities, we, however, leave this as a future research.

Next we study the average deviation between the requested $\theta^\ast$ and the allocated $\vartheta^\ast$, computed by the following formula
\begin{align} \label{E:avg_deviation}
    \Delta_{i}(t)=\frac{\sum_{i=1}^N\sum_{k=0}^t|(\vartheta^{i,\ast}_k-\theta^{i,\ast}_k)^\top \Delta|}{N(t+1)}.
\end{align}
We report the average deviation result for all three approaches in Fig. \ref{fig:group_avg_deviation_model}. 
The figure also shows that the average deviation is generally higher for the delay-insensitive approach compared to both delay-sensitive scenarios of reactive and impassive, confirming the explanations in the Section~ \ref{sec:weighted_cost}.

\begin{figure}[t]
    \centering
    \includegraphics[trim=10 0 10 20, clip,width=.935 \linewidth]{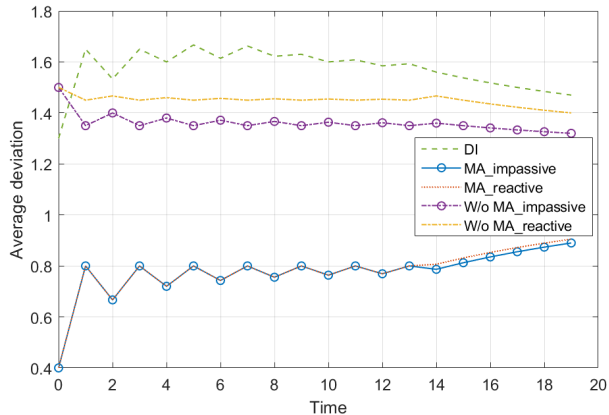}\vspace{-1mm}
    \caption{Average deviation in the allocated links as computed by \eqref{E:avg_deviation}.
    }
    \label{fig:group_avg_deviation_model}
    \vspace{-5.5mm}
\end{figure}

%% file: conclusions.tex
\section{CONCLUSION}\label{conclusion}\vspace{-1mm}

In this article, we address the problem of jointly optimal control and networking for multi-loop NCS exchanging data over a shared communication network that offers a range of capacity-limited, latency-varying and cost-prone transmission services. We investigate different awareness scenarios between the cross-layer decision makers and study the effects of the resulting interactions on the structure of the optimal policies. 
By formulating a system (social) optimization problem, we derive the joint optimal policies under various cross-layer awareness models of constant parameters and dynamic variables. We show that higher awareness leads to better social performance, however, results in more complex optimization problems. In addition, we discuss that tighter sensitivity w.r.t. the deviations from the desired local decision variables may lead to better local performance for certain systems, however, in a constrained setup where multiple systems are competing for limited resources, results in higher cost for other systems and eventually degrades the social performance. 
The proposed design approach is implemented on a multi-loop NCS where the simulation observations validate our theoretical results.

%% file: appendix.tex
\section*{Appendix}

\subsection{Proof of Theorem \ref{thm:CE}}\label{Append:thm1}

\begin{proof}
To compute $\vartheta_k^i$, the resource manager has no knowledge of $u_{[0,k-1]}^i$, but incorporates $\theta_{[0,k]}^i$'s, $i\!\in\!\mathrm{N}$, via $\tilde{\mathcal{I}}_k$. The controller $\mathcal{C}_i$ knows about~$u_{[0,k-1]}^i$, $\vartheta_{[0,k]}^i$ and $\theta_{[0,k]}^i$ via $\mathcal{I}_k^i$, while $u_{[0,k-1]}^i$, $\vartheta_{[0,k-1]}^i$ and $\theta_{[0,k-1]}^i$ are known for the delay controller via $\bar{\mathcal{I}}_k^i$. From (\ref{eq:local_objective})-(\ref{eq:global_OP}), we re-state (10) as
\begin{align}\label{prob:global_OP_outer}
&\min_{\gamma^i,\xi^i,\pi} \!J\!=\!\frac{1}{N}\sum\nolimits_{i=1}^N \E\bigg[\min_{\gamma^i,\pi}J^i(u^i,\vartheta^i)\;-\\\nonumber
&\min_{\gamma^i,\xi^i}\E\!\left[\|x_T^{i}\|_{Q_2^i}^2\!+\!\sum\nolimits_{k=0}^{T-1} \!\|x_k^{i}\|_{Q_1^i}^2\!+\!\|u_k^{i}\|_{R^i}^2\!+\theta_k^{i^\top}\!\Lambda\right]\!\bigg].
\end{align}
where, for the first term of (\ref{prob:global_OP_outer}), we obtain the following due to the one-directional independence of $\vartheta_k^i$ from $u_k^i$
\begin{align}\nonumber
&J^i(u^i,\vartheta^i)\!=\!\E\left[\E\!\left[\sum\nolimits_{k=0}^{T-1} \!\vartheta_{\!k}^{i^\top}\!\!\!\Lambda\Big|\tilde{\mathcal{I}}_k\!\right]\right]+\\\nonumber
&\E\left[\E\!\left[\|x_T^{i}\|_{Q_2^i}^2\!+\!\!\sum\nolimits_{k=0}^{T-1} \!\|x_k^{i}\|_{Q_1^i}^2\!+\!\|u_k^{i}\|_{R^i}^2\!\Big|\mathcal{I}_k^i,\tilde{\mathcal{I}}_k\!\right]\right].
\end{align}
We define $V_k^i=\|x_T^{i}\|_{Q_2^i}^2\!+\!\sum_{t=k}^{T-1} \!\|x_t^{i}\|_{Q_1^i}^2\!+\!\|u_t^{i}\|_{R^i}^2$. Since $\gamma^i$ is a local policy and its decision outcome $u^i$ is independent of all sub-systems $j\!\neq \!i$, and moreover, $\pi$ is independent of all $\gamma_i$'s, the optimal cost-to-go can be expressed as
\begin{align}
\min_{\substack{\gamma^i_{[k,T-1]}\\\pi_{[k,T-1]}}}J^i(u^i,\vartheta^i)=\!&\min_{\pi_{[k,T-1]}}\E\!\bigg[\min_{\gamma^i_{[k,T-1]}} \E\left[V_k^i\big|\mathcal{I}_k^i\right]+\!\!\\\nonumber
&\min_{\pi_{[k,T-1]}}\E\!\Big[\sum\nolimits_{t=k}^{T-1} \!\vartheta_t^{i^\top}\!\Lambda\big|\tilde{\mathcal{I}}_k\Big]\Big|\tilde{\mathcal{I}}_k\bigg]
\end{align}
For $J^i(u^i,\theta^i)$, we know $\bar{\mathcal{I}}_k^i\!\subseteq \!\mathcal{I}_k^i$, $\forall k$, from (\ref{set:controller-information}) and (\ref{set:DC-information}). Moreover, $u_k^i$ and $\theta_k^i$ are measurable w.r.t. $\mathcal{I}_k^i$ and $\bar{\mathcal{I}}_k^i$, respectively. Therefore, employing the tower property\footnote{For a random variable $X$ defined on a probability space with sigma-algebra $\mathcal{F}$, if $\E[X]\!<\!\infty$, then for any two sub-sigma-algebras $\mathcal{F}_1\!\subseteq \!\mathcal{F}_2\!\subseteq\! \mathcal{F}$, $\E[\E[X|\mathcal{F}_2]|\mathcal{F}_1]\!=\!\E[X|\mathcal{F}_1]$ \textit{almost surely}.}, and also using the law of total expectation\footnote{If the random variable $X$ is $\mathcal{F}$-measurable, then $\E[\E[X|\mathcal{F}]]=\E[X]$.}, we re-write (\ref{eq:local_objective}) as 
\begin{align*}
&J^i(u^i,\theta^i)\!=\\
&\E\!\left[\E\!\left[\E\!\left[\|x_T^{i}\|_{Q_2^i}^2\!+\!\sum\nolimits_{k=0}^{T-1} \!\|x_k^{i}\|_{Q_1^i}^2\!+\!\|u_k^{i}\|_{R^i}^2\!+\theta_k^{i^\top}\!\Lambda\Big|\mathcal{I}_k^i\right]\Big|\bar{\mathcal{I}}_k^i\!\right]\right].
\end{align*}
Hence, introducing $C_k^i(u^i,\theta^i)=V_k^i+\sum_{t=k}^{T-1}\theta_t^{i^\top}\!\Lambda$, we obtain
\begin{align*}
\min_{\substack{\gamma^i_{[k,T-1]}\\\xi^i_{[k,T-1]}}}\!\!J^i(u^i,\theta^i)\!=\!\E\!\bigg[\min_{\xi_{[k,T-1]}^i}\!\E\!\bigg[\!\min_{\gamma_{[k,T-1]}^i}\!\!\E\!\left[C_k^i(u^i,\theta^i)|\mathcal{I}_k^i\right]\bigg|\bar{\mathcal{I}}_k^i\bigg]\bigg]
\end{align*}
Finally, we can re-express (\ref{prob:global_OP_outer}) as
\begin{align}\label{eq:global_cost_separation}
\min_{\gamma^i,\xi^i,\pi} J&=\!\frac{1}{N}\sum\nolimits_{i=1}^N\E\Bigg\{\!\!\min_{\pi}\E\!\left[\min_{\gamma^i} \E\left[V_0^i\big|\mathcal{I}_0^i\right]\Big|\tilde{\mathcal{I}}_0\right]\\\nonumber
&+\min_{\pi}\E\!\bigg[\sum\nolimits_{k=0}^{T-1} \!\vartheta_k^{i^\top}\!\Lambda\Big|\tilde{\mathcal{I}}_0\bigg]\\\nonumber
&-\!\min_{\xi^i}\E\!\bigg[\min_{\gamma^i}\E\!\bigg[V_0^i\!+\!\sum\nolimits_{k=0}^{T-1} \!\theta_k^{i^\top}\!\Lambda\Big|\mathcal{I}_0^i\bigg]\Big|\bar{\mathcal{I}}_0^i\bigg]\!\Bigg\}.
\end{align}
The sole $\gamma^i$-dependent term in the above expression is $\E[V_0^i|\mathcal{I}_0^i]$, and since this term is minimized only by the control law $\gamma^i$, it coincides with the standard LQG problem. Therefore, for all $k\in[0,T-1]$, the following  control law  solves the inner optimization problem $\min_{\gamma^i} \E\left[V_0^i|\mathcal{I}_0^i\right]$
\begin{align}\label{eq:cost-to-go}
u_{[k,T-1]}^{i,\ast}&=\gamma^{i,\ast}_{[k,T-1]}(\mathcal{I}_k^i)=\argmin_{\gamma^i_{[k,T-1]}} \E\left[V_k^{i}|\mathcal{I}_k^i\right]\\\nonumber
&=\argmin_{\gamma^i_{[k,T-1]}}\E\!\left[\|x_T^{i}\|_{Q_2^i}^2\!+\!\!\sum\nolimits_{t=k}^{T-1} \!\|x_t^{i}\|_{Q_1^i}^2\!+\!\|u_t^{i}\|_{R^i}^2\big|\mathcal{I}_k^i\right]\!.
\end{align}
As (\ref{eq:cost-to-go}) is a standard LQG problem, we drop the derivation of $\gamma^{i,\ast}$ for brevity. This is, however, known that the optimal law $\gamma_k^{i,\ast}$ and gain $L_k^{i,\ast}$ in (\ref{eq:CE-law}) and (\ref{eq:CE-gain}) are the solutions of the problem (\ref{eq:cost-to-go}). (Full derivation can be found in \cite{8405590}.) 
\end{proof}

\subsection{Proof of Theorem \ref{thm:impassiveDC}}\label{Append:thm2}

\begin{proof}
The two assumptions on the independence of $\pi_k$ from $\gamma_k^i$'s, $i\!\in\!\mathrm{N}$, and $\bar{\mathcal{I}}_k^i\!\subseteq \!\mathcal{I}_k^i$ hold, so we begin from (\ref{eq:global_cost_separation}).
Recall that $\vartheta_{[0,k-1]}^i\notin\bar{\mathcal{I}}_k^i$, hence, to decide $\theta_k^i$, the delay controller presumes that the control signal is generated according to $\theta_{[0,k-1]}^i$ not $\vartheta_{[0,k-1]}^i$. We derived the optimal control policy that minimizes the sole $\gamma^i$-dependent term $V^i_0$ in (\ref{eq:global_cost_separation}), therefore, the optimal impassive delay control policy $\xi_k^{i,\ast}(\bar{\mathcal{I}}_k^i)$ will be obtained simply by minimizing the local LQG cost function $J^i(u^{i,\ast},\theta^i)$, i.e., $\forall k\in[0,T-1]$
\begin{equation}\label{eq:proof-delay-control}
\!\theta_{[k,T-1]}^{i,\ast}\!=\argmin_{\xi_{[k,T-1]}^i}\E\!\left[V^{i,\ast}_k(\gamma^{i,\ast}\!,\xi^i)\!+\!\sum\nolimits_{t=k}^{T-1}\theta_t^{i^\top}\!\!\Lambda\big|\bar{\mathcal{I}}_k^i\right]\!.\!
\end{equation}
Recalling Remark 2, we compute $V^{i,\ast}_k(\gamma^{i,\ast},\xi^i)$ at the impassive delay controller side. From the estimator dynamics (\ref{eq:estimator_dynamics}) and system dynamics (\ref{eq:sys_model}), the estimation error $e_k^i$ evolves as
\begin{equation*}
e_k^i=\sum\nolimits_{l=1}^{\tau_k^i}\sum\nolimits_{j=l}^{\tau_k^i}\bar{b}_{j,k}^iA_i^{l-1}w_{k-l}^i,
\end{equation*}
where $b_{j,k}^i$ in (\ref{eq:estimator_dynamics}) is replaced by $\bar{b}_{j,k}^i$ because the delay controller has no knowledge about the variables $\{\vartheta_0^i,\ldots,\vartheta_{k-1}^i\}$ (the plant controller and the collocated estimator have this knowledge). Since $\bar{\mathcal{I}}_k^i\subseteq \mathcal{I}_k^i$, it is, moreover, straightforward to compute $\E[\E[e_k^ie_k^{i^\top}|\mathcal{I}_k^i]|\bar{\mathcal{I}}_k^i]=\E[e_k^ie_k^{i^\top}|\bar{\mathcal{I}}_k^i]$, as follows:
\begin{align*}
\E[e_k^ie_k^{i^\top}\big|\bar{\mathcal{I}}_k^i]&=\sum\nolimits_{l=1}^{\tau_k^i}\sum\nolimits_{j=l}^{\tau_k^i}\bar{b}_{j,k}^i\E[A_i^{l-1}w_{k-l}^i w_{k-l}^{i^\top}A_i^{{l-1}^\top}]\\
&=\sum\nolimits_{l=1}^{\tau_k^i}\sum\nolimits_{j=l}^{\tau_k^i}\bar{b}_{j,k}^iA_i^{l-1}\Sigma_{k-l}^i A_i^{{l-1}^\top},
\end{align*} 
where, $\Sigma_{k-l}^i\!=\!\Sigma_{x_0}^i$, $k\!<\!l$, and $\Sigma_{k-l}^i=\Sigma_{w}^i$, $k\geq l$. Having this and noting that $\bar{\mathcal{I}}_0^i=\{A_i, B_i, Q_1^i, Q_2^i, R^i, \Sigma_w^i,\Sigma_{x_0}^i\}$, we can rewrite $\E[V^{i,\ast}_0(\gamma^{i,\ast},\xi^i)|\bar{\mathcal{I}}_0^i]$ as follows
\begin{align}\label{eq:proof-optimal-value-function}
\E&[V^{i,\ast}_0(\gamma^{i,\ast},\xi^i)|\bar{\mathcal{I}}_0^i]=\|\!\E\left[x_0^i\right]\!\|^2_{P_0^i}+\!\sum\nolimits_{t=1}^T \!\Tr (P_t^i \Sigma_w^i)\\\nonumber
&+\Tr(P_0^i\sum\nolimits_{l=1}^{\tau_0^i}\sum\nolimits_{j=l}^{\tau_0^i}\bar{b}_{j,0}^iA_i^{{l-1}^\top}\Sigma_{x_0}^i A_i^{l-1})\\\nonumber
&+\sum\nolimits_{t=0}^{T-1}\Tr(\tilde{P}_t^i \sum\nolimits_{l=1}^{\tau_t^i}\sum\nolimits_{j=l}^{\tau_t^i}\bar{b}_{j,t}^iA_i^{{l-1}^\top}\Sigma_{t-l}^i A_i^{l-1}).
\end{align}
As the only term in the expression above that is dependent on $\theta_{[0,T-1]}^i$ is the last term, the optimization problem (\ref{eq:proof-delay-control}) can equivalently be expressed, initiating from the time $k=0$, as
\begin{align*}
&\theta_{[0,T-1]}^{i,\ast}\!=\argmin_{\xi_{[0,T-1]}^i}\E\left[V^{i,\ast}_0(\gamma^{i,\ast},\xi^i)+\sum\nolimits_{t=0}^{T-1}\theta_t^{i^\top}\Lambda\big|\bar{\mathcal{I}}_0^i\right]=\\
&\argmin_{\xi_{[0,T-1]}^i}\sum_{t=0}^{T-1}\!\left[\Tr(\tilde{P}_t^i \sum_{l=1}^{\tau_t^i}\sum_{j=l}^{\tau_t^i}\bar{b}_{j,t}^iA_i^{{l-1}^\top}\!\Sigma_{t-l}^i A_i^{l-1})+\theta_t^{i^\top}\!\Lambda\right]
\end{align*}
The constraints of the problem (\ref{eq:opt-imp-del-cont}) are all linear and $\theta_k^i$ is binary-valued, hence the above problem is a MILP. 
Moreover, it is independent from both the noise realizations and $\vartheta_{[0,T-1]}$, thus $\theta_{[0,T-1]}^\ast$ can be computed offline. The constraint $\sum_{l=0}^D\theta_t^i(l)\!=\!1$ ensures that only one delay link is selected per-time, while the last two constraints look after convenient indexes for $\bar{b}_{j,k}^i$ for $k\!\geq \!D$ and $k\!<\!D$ (see the Corollary 1). 

To find $\pi^{\ast}$, we use a similar procedure to that of computing $\xi^{i,\ast}$, except $\vartheta_k^i$ is now computed knowing the information $\{\theta_{[0,k]}^{i,\ast},\vartheta_{[0,k-1]}^{i,\ast}\}$, $\forall i$. We compute $\E[V^{i,\ast}_0(\gamma^{i,\ast},\pi)|\tilde{\mathcal{I}}_0]$ that results in a similar expression as on the right side of the equality in (\ref{eq:proof-optimal-value-function}) with the exception being $\bar{b}_{j,t}^i$ replaced by $b_{j,t}^i$. Hence, from (\ref{eq:global_cost_separation}), and considering the resource constraint $\sum_{i=1}^N \vartheta_t^i(d)\leq c_d, \forall d\in\mathcal{D}$, and the latency deviation constraint $-\alpha_i\!\leq\!(\vartheta_t^i-\theta_t^{i,\ast})^\top \Delta\!\leq \beta_i, i\in \mathrm{N}$, we derive the optimal resource allocation offline from the following MILP:
\begin{align*}
&\vartheta_{[k,T-1]}^{\ast}\!=\argmin_{\pi_{[k,T-1]}}\frac{1}{N}\sum_{i=1}^N \E\!\bigg[V^{i,\ast}_k(\gamma^{i,\ast}\!,\pi^i)\!+\!\!\sum_{t=k}^{T-1}\vartheta_t^{i^\top}\!\Lambda\big|\tilde{\mathcal{I}}_k\bigg]\!=\\
&\argmin_{\pi_{[k,T-1]}}\frac{1}{N}\!\sum_{i=1}^N\sum_{t=k}^{T-1}\!\bigg[\vartheta_t^{i^\top} \!\!\Lambda\!+\!\sum_{l=0}^{\tau_{t}^i} \sum_{j=l}^{\tau_{t}^i} b_{j,t}^i \textsf{Tr}(\tilde{P}_{t}^i A_{i}^{{l-1}^{\textsf{T}}} \Sigma_{w}^i A_{i}^{l-1})\bigg]\!
\end{align*}
Since $\theta^{i,\ast}_{[0,T-1]}$ is computed offline from (\ref{eq:opt-imp-del-cont}) independent of $\vartheta^{i}_{[0,T-1]}$, we can set $k=0$ above to complete the proof.
\end{proof}

\subsection{Proof of Corollary \ref{corol:performance-comparison}}\label{Append:corol3}
\begin{proof}
The control policy $\gamma^{i,\ast}$ follows (\ref{eq:CE-law}) for both impassive and reactive scenarios, so we only compare the optimal cost values of the joint policies $(\xi^{i,\ast},\pi^\ast)$ derived from Theorems \ref{thm:impassiveDC} and \ref{thm:reactiveDC}. Define $(\bar{\theta}^{i,\ast},\bar{\vartheta}^{i,\ast})$ and $(\tilde{\theta}^{i,\ast},\tilde{\vartheta}^{i,\ast})$, respectively, as the joint optimal impassive and reactive delay control and resource allocation variables over time horizon $[0,T]$. First assume $\bar{\theta}^{i,\ast}\!=\!\tilde{\theta}^{i,\ast}$, then $\bar{b}_{j,t}^i\!=\!\tilde{b}_{j,t}^i, \!\!\; \forall t$ must hold from (\ref{eq:opt-imp-del-cont}) and (\ref{eq:thm3-del-ctrl}), which leads to $\bar{\vartheta}^{i,\ast}\!=\!\tilde{\vartheta}^{i,\ast}$ from (\ref{eq:opt-imp-res-alloc}) and (\ref{eq:thm3-res-manager}). 
Having the problems (\ref{eq:opt-imp-del-cont}) and (\ref{eq:thm3-del-ctrl}), and also (\ref{eq:opt-imp-res-alloc}) and (\ref{eq:thm3-res-manager}) coincide, it easily leads to $J^{i,\ast}_{\text{Re}}\!=\! J^{i,\ast}_{\text{Im}}$ and $J^\ast_{\text{Re}}\!=\! J^\ast_{\text{Im}}$. 

Now assume $\bar{\theta}^{i,\ast}\!\neq\!\tilde{\theta}^{i,\ast}$. Due to the fact that the information set $\bar{\mathcal{I}}^i_{[0,T-1]}$ associated with the impassive approach (given in (\ref{set:DC-information-impassive})) is a subset of its counterpart associated with the reactive approach (given in (\ref{set:DC-information})), any optimal solution of the problem (\ref{eq:opt-imp-del-cont}) can also be obtained from the problem (\ref{eq:thm3-del-ctrl}) if it is optimal for the latter.
Hence, if $\bar{\theta}^{i,\ast}\!\neq\!\tilde{\theta}^{i,\ast}$, then $\bar{\theta}^{i,\ast}$ is not the optimal solution of  problem (\ref{eq:thm3-del-ctrl}), which implies $J^{i,\ast}_{\text{Re}}(u^{i,\ast}\!,\tilde{\theta}^{i,\ast})\!<\!J^{i,\ast}_{\text{Im}}(u^{i,\ast}\!,\bar{\theta}^{i,\ast})$.
For the resource allocation, assume $\tilde{\vartheta}^{i,\ast}$ be the optimal solution of the problem (\ref{eq:thm3-res-manager}) such that $\tilde{\vartheta}^{i,\ast}\!\neq\!\bar{\vartheta}^{i,\ast}$ while $J^\ast_{\text{Re}}\!>\! J^\ast_{\text{Im}}$. 
Recall that $\bar{\vartheta}^{i,\ast}$ is the optimal resource allocation in response to $\bar{\theta}^{i,\ast}$ computed from (\ref{eq:opt-imp-del-cont}), while we know if $\tilde{\vartheta}^{i,\ast}\!\neq\!\bar{\vartheta}^{i,\ast}$, then $\bar{\theta}^{i,\ast}\!\neq\!\tilde{\theta}^{i,\ast}$. 
Knowing this, together with $J^\ast_{\text{Re}}\!>\! J^\ast_{\text{Im}}$, implies that the joint policy $(\bar{\theta}^{i,\ast}\!,\bar{\vartheta}^{i,\ast})$ outperforms $(\tilde{\theta}^{i,\ast}\!,\tilde{\vartheta}^{i,\ast})$, which requires $J^{i,\ast}_{\text{Re}}(u^{i,\ast}\!,\tilde{\theta}^{i,\ast})\!>\!J^{i,\ast}_{\text{Im}}(u^{i,\ast}\!,\bar{\theta}^{i,\ast})$ to hold. 
This, however, contradicts the previous condition ensuring that if $\bar{\theta}^{i,\ast}\!\neq\!\tilde{\theta}^{i,\ast}$, then $J^{i,\ast}_{\text{Re}}(u^{i,\ast}\!,\tilde{\theta}^{i,\ast})\!<\!J^{i,\ast}_{\text{Im}}(u^{i,\ast}\!,\bar{\theta}^{i,\ast})$, and hence the condition $J^\ast_{\text{Re}}\!>\! J^\ast_{\text{Im}}$ cannot be realized if $\tilde{\vartheta}^{i,\ast}\!\neq\!\bar{\vartheta}^{i,\ast}$.
\end{proof}